\documentclass[12pt]{article}
\usepackage{amsmath, amssymb,amsfonts,bbm,verbatim,multirow,color}
\usepackage{epic,eepic,psfrag,epsfig}
\usepackage[sf,bf,SF,footnotesize]{subfigure}
\usepackage{graphicx}
\usepackage{amsthm,breakcites}
\usepackage{amscd}
\usepackage{epsfig}
\usepackage{fullpage}
\usepackage{natbib}
\usepackage{setspace}
\usepackage{enumerate}
\usepackage{array}
\usepackage[small]{caption}
\RequirePackage[colorlinks,citecolor=blue,urlcolor=blue,breaklinks]{hyperref}

\newtheorem{thm}{Theorem}

\newtheorem{lemma}[thm]{Lemma}

\DeclareMathOperator*{\sargmin}{sargmin}
\DeclareMathOperator*{\argmin}{argmin}

\DeclareMathOperator*{\Var}{Var}

\def\hat{\widehat}

\begin{document}
\title{Nonparametric independence testing via mutual information}
\author{Thomas B. Berrett and Richard J. Samworth\\Statistical Laboratory, University of Cambridge}

\maketitle

\begin{abstract}
We propose a test of independence of two multivariate random vectors, given a sample from the underlying population. Our approach, which we call \texttt{MINT}, is based on the estimation of mutual information, whose decomposition into joint and marginal entropies facilitates the use of recently-developed efficient entropy estimators derived from nearest neighbour distances.  The proposed critical values, which may be obtained from simulation (in the case where one marginal is known) or resampling, guarantee that the test has nominal size, and we provide local power analyses, uniformly over classes of densities whose mutual information satisfies a lower bound. Our ideas may be extended to provide a new goodness-of-fit tests of normal linear models based on assessing the independence of our vector of covariates and an appropriately-defined notion of an error vector.  The theory is supported by numerical studies on both simulated and real data.
\end{abstract}

\section{Introduction}

Independence is a fundamental concept in statistics and many related fields, underpinning the way practitioners frequently think about model building, as well as much of statistical theory.  Often we would like to assess whether or not the assumption of independence is reasonable, for instance as a method of exploratory data analysis \citep{Steuer02,Albert15,Nguyen17}, or as a way of evaluating the goodness-of-fit of a statistical model \citep{EinmahlvanKeilegom2008}.  Testing independence and estimating dependence are well-established areas of statistics, with the related idea of the correlation between two random variables dating back to Francis Galton's work at the end of the 19th century \citep{Stigler89}, which was subsequently expanded upon by Karl Pearson \citep[e.g.][]{Pearson1920}. Since then many new measures of dependence have been developed and studied, each with its own advantages and disadvantages, and there is no universally accepted measure. For surveys of several measures, see, for example, \citet{Schweizer81}, \citet{Joe89}, \citet{Mari01} and the references therein. We give an overview of more recently-introduced quantities below; see also \citet{Josse14}.

In addition to the applications mentioned above, dependence measures play an important role in independent component analysis (ICA), a special case of blind source separation, in which a linear transformation of the data is sought so that the transformed data is maximally independent; see e.g.\ \citet{Comon94}, \citet{Bach02}, \citet{Miller:03} and \citet{Samworth:12b}.  Here, independence tests may be carried out to check the convergence of an ICA algorithm and to validate the results \citep[e.g.][]{Wu09}.  Further examples include feature selection, where one seeks a set of features which contains the maximum possible information about a response \citep{Torkkola03, Song12}, and the evaluation of the quality of a clustering in cluster analysis \citep{Vinh10}.

When dealing with discrete data, often presented in a contingency table, the independence testing problem is typically reduced to testing the equality of two discrete distributions via a chi-squared test.  Here we will focus on the case of distributions that are absolutely continuous with respect to the relevant Lebesgue measure. Classical nonparametric approaches to measuring dependence and independence testing in such cases include Pearson's correlation coefficient, Kendall's tau and Spearman's rank correlation coefficient. Though these approaches are widely used, they suffer from a lack of power against many alternatives; indeed Pearson's correlation only measures linear relationships between variables, while Kendall's tau and Spearman's rank correlation coefficient measure monotonic relationships.  Thus, for example, if $X$ has a symmetric distribution on the real line, and $Y=X^2$, then the population quantities corresponding to these test statistics are zero in all three cases.  Hoeffding's test of independence \citep{Hoeffding48} is able to detect a wider class of departures from independence and is distribution-free under the null hypothesis but, as with these other classical methods, was only designed for univariate variables.  Recent work of \citet{Weihs17} has aimed to address some of computational challenges involved in extending these ideas to multivariate settings. 

Recent research has focused on constructing tests that can be used for more complex data and that are consistent against wider classes of alternatives. The concept of distance covariance was introduced in \citet{Szekely07} and can be expressed as a weighted $L_2$ norm between the characteristic function of the joint distribution and the product of the marginal characteristic functions. This concept has also been studied in high dimensions \citep{Szekely13,YZS2017}, and for testing independence of several random vectors \citep{Micheaux2017}. In \citet{Sejdinovic13} tests based on distance covariance were shown to be equivalent to a reproducing kernel Hilbert space (RKHS) test for a specific choice of kernel. RKHS tests have been widely studied in the machine learning community, with early understanding of the subject given by \citet{Bach02} and \citet{Gretton05}, in which the Hilbert--Schmidt independence criterion was proposed. These tests are based on embedding the joint distribution and product of the marginal distributions into a Hilbert space and considering the norm of their difference in this space. One drawback of the kernel paradigm here is the computational complexity, though \citet{Jitkrittum16} and \citet{Zhang17} have recently attempted to address this issue. The performance of these methods may also be strongly affected by the choice of kernel.  In another line of work, there is a large literature on testing independence based on an empirical copula process; see for example \citet{Kojadinovic2009} and the references therein. Other test statistics include those based on partitioning the sample space \citep[e.g.][]{Gretton10,Heller16}. These have the advantage of being distribution-free under the null hypothesis, though their performance depends on the particular partition chosen. 

We also remark that the basic independence testing problem has spawned many variants.  For instance, \citet{Pfister17} extend kernel tests to tests of mutual independence between a group of random vectors.  Another important extension is to the problem of testing conditional independence, which is central to graphical modelling \citep{Lauritzen96} and also relevant to causal inference \citep{Pearl2009}.  Existing tests of conditional independence include the proposals of \citet{Su08}, \citet{Zhang11} and \citet{Fan17}.  

To formalise the problem we consider, let $X$ and $Y$ have (Lebesgue) densities $f_X$ on $\mathbb{R}^{d_X}$ and $f_Y$ on $\mathbb{R}^{d_Y}$ respectively, and let $Z = (X,Y)$ have density $f$ on $\mathbb{R}^d$, where $d:= d_X+d_Y$.  Given independent and identically distributed copies $Z_1, \ldots, Z_n$ of $Z$, we wish to test the null hypothesis that $X$ and $Y$ are independent, denoted $H_0: X \perp\!\!\!\perp Y$, against the alternative that $X$ and $Y$ are not independent, written $H_1: X \not\!\perp\!\!\!\perp Y$.  Our approach is based on constructing an estimator $\hat{I}_n = \hat{I}_n(Z_1,\ldots,Z_n)$ of the mutual information $I(X;Y)$ between $X$ and $Y$.  Mutual information turns out to be a very attractive measure of dependence in this context; we review its definition and basic properties in Section~\ref{Sec:Mutual} below.  In particular, its decomposition into joint and marginal entropies (see~\eqref{Eq:Decomp} below) facilitates the use of recently-developed efficient entropy estimators derived from nearest neighbour distances \citep{BSY2017}.

The next main challenge is to identify an appropriate critical value for the test.  In the simpler setting where either of the marginals $f_X$ and $f_Y$ is known, a simulation-based approach can be employed to yield a test of any given size $q \in (0,1)$.  We further provide regularity conditions under which the power of our test converges to $1$ as $n \rightarrow \infty$, uniformly over classes of alternatives with $I(X;Y) \geq b_n$, say, where we may even take $b_n = o(n^{-1/2})$.  To the best of our knowledge this is the first time that such a local power analysis has been carried out for an independence test for multivariate data.  When neither marginal is known, we obtain our critical value via a permutation approach, again yielding a test of the nominal size.  Here, under our regularity conditions, the test is uniformly consistent (has power converging to 1) against alternatives whose mutual information is bounded away from zero.  We call our test \texttt{MINT}, short for \underline{M}utual \underline{In}formation \underline{T}est; it is implemented in the \texttt{R} package \textbf{IndepTest} \citep{BGS2017}.

As an application of these ideas, we are able to introduce new goodness-of-fit tests of normal linear models based on assessing the independence of our vector of covariates and an appropriately-defined notion of an error vector.  Such tests do not follow immediately from our earlier work, because we do not observe realisations of the error vector directly; instead, we only have access to residuals from a fitted model.  Nevertheless, we are able to provide rigorous justification, again in the form of a local analysis, for our approach.  It seems that, when fitting normal linear models, current standard practice in the applied statistics community for assessing goodness-of-fit is based on visual inspection of diagnostic plots such as those provided by the \texttt{plot} command in \texttt{R} when applied to an object of class \texttt{lm}.  Our aim, then, is to augment the standard toolkit by providing a formal basis for inference regarding the validity of the model.  Related work here includes \citet{Neumeyer2009}, \citet{NeumeyervanKeilegom2010}, \citet{MSW2012}, \citet{SenSen2014} and \citet{ShahBuhlmann2017}.

The remainder of the paper is organised as follows: after reviewing the concept of mutual information in Section~\ref{Sec:Defn}, we explain in Section~\ref{Sec:MIEstimation} how it can be estimated effectively from data using efficient entropy estimators.  Our new tests, for the cases where one of the marginals is known and where they are both unknown, are introduced in Sections~\ref{Sec:KnownMarginals} and~\ref{Sec:UnknownMarginals} respectively.  The regression setting is considered in Section~\ref{Sec:Regression}, and numerical studies on both simulated and real data are presented in Section~\ref{Sec:Simulations}.  Proofs are given in Section~\ref{Sec:Proofs}.  

The following notation is used throughout.  For a generic dimension $D \in \mathbb{N}$, let $\lambda_D$ and $\|\cdot\|$ denote Lebesgue measure and the Euclidean norm on $\mathbb{R}^D$ respectively.  If $f = d\mu/d\lambda_D$ and $g = d\nu/d\lambda_D$ are densities on $\mathbb{R}^D$ with respect to $\lambda_D$, we write $f \ll g$ if $\mu \ll \nu$.  For $z \in \mathbb{R}^D$ and $r \in [0,\infty)$, we write $B_z(r) := \{w \in \mathbb{R}^D:\|w-z\| \leq r\}$ and $B_z^\circ(r) := B_z(r) \setminus \{z\}$.  We write $\lambda_{\min}(A)$ for the smallest eigenvalue of a positive definite matrix $A$, and $\|B\|_{\mathrm{F}}$ for the Frobenius norm of a matrix $B$.

\section{Mutual information}
\label{Sec:Mutual}

\subsection{Definition and basic properties}
\label{Sec:Defn}

Retaining our notation from the introduction, let $\mathcal{Z} := \{(x,y):f(x,y) > 0\}$.  A very natural measure of dependence is the mutual information between $X$ and $Y$, defined to be
\[
	I(X;Y)= I(f) := \int_{\mathcal{Z}} f(x,y) \log \frac{f(x,y)}{f_X(x)f_Y(y)} \,d\lambda_d(x,y),
\]
when $f \ll f_Xf_Y$, and defined to be $\infty$ otherwise.  This is the Kullback--Leibler divergence between the joint distribution of $(X,Y)$ and the product of the marginal distributions, so is non-negative, and equal to zero if and only if $X$ and $Y$ are independent.  
Another attractive feature of mutual information as a measure of dependence is that it is invariant to smooth, invertible transformations of $X$ and $Y$. Indeed, suppose that $\mathcal{X} := \{x:f_X(x) > 0\}$ is an open subset of $\mathbb{R}^{d_X}$ and $\phi:\mathcal{X} \rightarrow \mathcal{X}'$ is a continuously differentiable bijection whose inverse $\phi^{-1}$ has Jacobian determinant $J(x')$ that never vanishes on $\mathcal{X}'$.  Let $\Phi:\mathcal{X} \times \mathcal{Y} \rightarrow \mathcal{X}' \times \mathcal{Y}$ be given by $\Phi(x,y) := \bigl(\phi(x),y\bigr)$, so that $\Phi^{-1}(x',y) = \bigl(\phi^{-1}(x'),y)$ also has Jacobian determinant $J(x')$.  Then $\Phi(X,Y)$ has density $g(x',y) = f\bigl(\phi^{-1}(x'),y\bigr)|J(x')|$ on $\Phi(\mathcal{Z})$ and $X' = \phi(X)$ has density $g_{X'}(x') = f_X\bigl(\phi^{-1}(x')\bigr)|J(x')|$ on $\mathcal{X}'$.  It follows that 
\begin{align*}
	I\bigl(\phi(X);Y\bigr) &= \int_{\Phi(\mathcal{Z})} g(x',y) \log \frac{g(x',y)}{g_{X'}(x')f_Y(y)} \, d\lambda_d(x',y) \\
&= \int_{\Phi(\mathcal{Z})} g(x',y) \log \frac{f\bigl(\phi^{-1}(x'),y\bigr)}{f_X\bigl(\phi^{-1}(x')\bigr)f_Y(y)} \, d\lambda_d(x',y) = I(X;Y).
\end{align*}
This means that mutual information is \emph{nonparametric} in the sense of \citet{Weihs17}, whereas several other measures of dependence, including distance covariance, RKHS measures and correlation-based notions are not in general.  Under a mild assumption, the mutual information between $X$ and $Y$ can be expressed in terms of their joint and marginal entropies; more precisely, writing $\mathcal{Y} := \{y:f_Y(y) > 0\}$, and provided that each of $H(X,Y)$, $H(X)$ and $H(Y)$ are finite,
\begin{align}
\label{Eq:Decomp}
I(X;Y) &= \int_{\mathcal{Z}} f \log f \, d\lambda_d - \int_{\mathcal{X}} f_X \log f_X \, d\mu_{d_X} - \int_{\mathcal{Y}} f_Y \log f_Y \, d\mu_{d_Y} \nonumber \\
&=: -H(X,Y) + H(X) + H(Y).
\end{align}
Thus, mutual information estimators can be constructed from entropy estimators.  

Moreover, the concept of mutual information is easily generalised to more complex situations.  For instance, suppose now that $(X,Y,W)$ has joint density $f$ on $\mathbb{R}^{d+d_W}$, and let $f_{(X,Y)|W}(\cdot|w)$, $f_{X|W}(\cdot|w)$ and $f_{Y|W}(\cdot|w)$ denote the joint conditional density of $(X,Y)$ given $W=w$ and the conditional densities of $X$ given $W=w$ and $Y$ given $W=w$ respectively.  When $f_{(X,Y)|W}(\cdot,\cdot|w) \ll f_{X|W}(\cdot|w)f_{Y|W}(\cdot|w)$ for each $w$ in the support of $W$, the conditional mutual information between $X$ and $Y$ given $W$ is defined as
\[
I(X;Y|W) := \int_{\mathcal{W}} f(x,y,w) \log \frac{f_{(X,Y)|W}(x,y|w)}{f_{X|W}(x|w)f_{Y|W}(y|w)} \, d\lambda_{d+d_W}(x,y,w),
\]
where $\mathcal{W} := \{(x,y,w):f(x,y,w) > 0\}$.  This can similarly be written as
\[
I(X;Y|W) = H(X,W) + H(Y,W) - H(X,Y,W) - H(W),
\]
provided each of the summands is finite.  

Finally, we mention that the concept of mutual information easily generalises to situations with $p$ random vectors.  In particular, suppose that $X_1,\ldots,X_p$ have joint density $f$ on $\mathbb{R}^d$, where $d = d_1 + \ldots + d_p$ and that $X_j$ has marginal density $f_j$ on $\mathbb{R}^{d_j}$.  Then, when $f \ll f_1\ldots f_p$, and writing $\mathcal{X}_p := \{(x_1,\ldots,x_p) \in \mathbb{R}^d:f(x_1,\ldots,x_p) > 0\}$, we can define
\begin{align*}
I(X_1;\ldots;X_p) &:= \int_{\mathcal{X}_p} f(x_1,\ldots,x_p) \log \frac{f(x_1,\ldots,x_p)}{f_1(x_1)\ldots f_p(x_p)} \, d\lambda_d(x_1,\ldots,x_p) \\
&= \sum_{j=1}^p H(X_j) - H(X_1,\ldots,X_p),
\end{align*}
with the second equality holding provided that each of the entropies is finite. The tests we introduce in Sections~\ref{Sec:KnownMarginals} and~\ref{Sec:UnknownMarginals} therefore extend in a straightforward manner to tests of independence of several random vectors. 

\subsection{Estimation of mutual information}
\label{Sec:MIEstimation}

For $i=1,\dots,n$, let $Z_{(1),i},\ldots,Z_{(n-1),i}$ denote a permutation of $\{Z_1,\ldots,Z_n\} \setminus \{Z_i\}$ such that $\|Z_{(1),i} - Z_i\| \leq \ldots \leq \|Z_{(n-1),i} - Z_i\|$.  For conciseness, we let
\[
\rho_{(k),i} := \|Z_{(k),i} - Z_i\|
\]
denote the distance between $Z_i$ and the $k$th nearest neighbour of $Z_i$.  To estimate the unknown entropies, we will use a weighted version of the Kozachenko--Leonenko estimator \citep{Kozachenko:87}.  For $k = k^Z \in \{1,\ldots,n-1\}$ and weights $w_1^Z,\ldots,w_k^Z$ satisfying $\sum_{j=1}^k w_j^Z = 1$, this is defined as 
\[
\hat{H}_n^Z = \hat{H}_{n,k}^{d,w^Z}(Z_1,\ldots,Z_n) := \frac{1}{n}\sum_{i=1}^n \sum_{j=1}^k w_j^Z \log \biggl(\frac{\rho_{(j),i}^d V_d (n-1)}{e^{\Psi(j)}}\biggr),
\]
where $V_d := \pi^{d/2}/\Gamma(1 + d/2)$ denotes the volume of the unit $d$-dimensional Euclidean ball and where $\Psi$ denotes the digamma function.  \citet{BSY2017} provided conditions on $k$, $w_1^Z,\ldots,w_k^Z$ and the underlying data generating mechanism under which $\hat{H}_n^Z$ is an efficient estimator of $H(Z)$ (in the sense that its asymptotic normalised squared error risk achieves the local asymptotic minimax lower bound) in arbitrary dimensions.  
With estimators $\hat{H}_n^X$ and $\hat{H}_n^Y$ of $H(X)$ and $H(Y)$ defined analogously as functions of $(X_1,\ldots,X_n)$ and $(Y_1,\ldots,Y_n)$ respectively, we can use~\eqref{Eq:Decomp} to define an estimator of mutual information by
\begin{equation}
\label{Eq:MIEst}
	\hat{I}_n = \hat{I}_n(Z_1, \ldots, Z_n)= \hat{H}_n^X + \hat{H}_n^Y - \hat{H}_n^Z.
\end{equation}
Having identified an appropriate mutual information estimator, we turn our attention in the next two sections to obtaining appropriate critical values for our independence tests.


\section{The case of one known marginal distribution}
\label{Sec:KnownMarginals}


In this section, we consider the case where at least one of $f_X$ and $f_Y$ in known (in our experience, little is gained by knowledge of the second marginal density), and without loss of generality, we take the known marginal to be $f_Y$.  We further assume that we can generate independent and identically distributed copies of $Y$, denoted $\{Y_i^{(b)}:i=1,\ldots,n,b=1,\ldots,B\}$, independently of $Z_1,\ldots,Z_n$.  Our test in this setting, which we refer to as \texttt{MINTknown} (or \texttt{MINTknown}$(q)$ when the nominal size $q \in (0,1)$ needs to be made explicit), will reject $H_0$ for large values of $\hat{I}_n$.  The ideal critical value, if both marginal densities were known,  would therefore be  
\[
	C_q^{(n)} := \inf \{ r \in \mathbb{R} : \mathbb{P}_{f_Xf_Y}( \hat{I}_n > r) \leq q \}.
\]
Using our pseudo-data $\{Y_i^{(b)}:i=1,\ldots,n,b=1,\ldots,B\}$, generated as described above, we define the statistics
\[
	\hat{I}_n^{(b)}= \hat{I}_n \bigl( (X_1, Y_1^{(b)}), \ldots, (X_n, Y_n^{(b)}) \bigr)
\]
for $b=1, \ldots, B$. Motivated by the fact that these statistics have the same distribution as $\hat{I}_n$ under $H_0$, we can estimate the critical value $C_q^{(n)}$ by
\[
	\hat{C}_q^{(n),B} := \inf \biggl\{ r \in \mathbb{R} : \frac{1}{B+1}\sum_{b=0}^B \mathbbm{1}_{\{\hat{I}_n^{(b)} \geq r\}} \leq q\biggr\},
\]
the $(1-q)$th quantile of $\{ \hat{I}_n^{(0)}, \ldots, \hat{I}_n^{(B)}\}$, where $\hat{I}_n^{(0)} := \hat{I}_n$.  The following lemma justifies this critical value estimate.
\begin{lemma}
\label{Lemma:Size}
For any $q \in (0,1)$ and $B \in \mathbb{N}$, the \emph{\texttt{MINTknown}}$(q)$ test that rejects $H_0$ if and only if $\hat{I}_n > \hat{C}_q^{(n),B}$ has size at most $q$, in the sense that
\[
\sup_{k \in \{1,\ldots,n-1\}} \sup_{(X,Y):I(X;Y) = 0} \mathbb{P}\bigl(\hat{I}_n > \hat{C}_q^{(n),B}\bigr) \leq q,
\]
where the inner supremum is over all joint distributions of pairs $(X,Y)$ with $I(X;Y) = 0$.
\end{lemma}
An interesting feature of \texttt{MINTknown}, which is apparent from the proof of Lemma~\ref{Lemma:Size}, is that there is no need to calculate $\hat{H}_n^X$ in~\eqref{Eq:Decomp}, either on the original data, or on the pseudo-data sets $\{(X_1, Y_1^{(b)}), \ldots, (X_n, Y_n^{(b)}):b = 1,\ldots,B\}$.  This is because in the decomposition of the event $\{\hat{I}_n^{(b)} \geq \hat{I}_n\}$ into entropy estimates, $\hat{H}_n^X$ appears on both sides of the inequality, so it cancels.  This observation somewhat simplifies our assumptions and analysis, as well as reducing the number of tuning parameters that need to be chosen.

The remainder of this section is devoted to a rigorous study of the power of \texttt{MINTknown} that is compatible with a sequence of local alternatives $(f_n)$ having mutual information $I_n \rightarrow 0$.  To this end, we first define the classes of alternatives that we consider; these parallel the classes introduced by \citet{BSY2017} in the context of entropy estimation.  Let $\mathcal{F}_d$ denote the class of all density functions with respect to Lebesgue measure on $\mathbb{R}^d$.  For $f \in \mathcal{F}_d$ and $\alpha > 0$, let
\[
\mu_\alpha(f) := \int_{\mathbb{R}^d} \|z\|^\alpha f(z) \, dz.
\]
Now let $\mathcal{A}$ denote the class of decreasing functions $a:(0,\infty) \rightarrow [1,\infty)$ satisfying $a(\delta) = o(\delta^{-\epsilon})$ as $\delta \searrow 0$, for every $\epsilon > 0$.  If $a \in \mathcal{A}$, $\beta > 0$, $f \in \mathcal{F}_d$ is $m := (\lceil \beta \rceil -1$)-times differentiable and $z \in \mathcal{Z}$, we define $r_a(z):=\{8d^{1/2}a(f(z))\}^{-1/(\beta \wedge 1)}$ and
\[
M_{f,a,\beta}(z) := \max \biggl\{ \max_{t=1,\ldots, m} \frac{\|f^{(t)}(z)\|}{f(z)} \, , \, \sup_{w \in B_z^\circ(r_a(z))} \frac{\|f^{(m)}(w)-f^{(m)}(z)\|}{f(z) \|w-z\|^{\beta- m}} \biggr\}.
\]
The quantity $M_{f,a,\beta}(z)$ measures the smoothness of derivatives of $f$ in neighbourhoods of $z$, relative to $f(z)$ itself.  Note that these neighbourhoods of $z$ are allowed to become smaller when $f(z)$ is small.  Finally, for $\Theta := (0,\infty)^4 \times \mathcal{A}$, and $\theta = (\alpha,\beta,\nu,\gamma,a) \in \Theta$, let
\[
\mathcal{F}_{d,\theta} := \biggl\{f \in \mathcal{F}_d: \mu_\alpha(f) \leq \nu, \|f\|_\infty \leq \gamma, \sup_{z:f(z) \geq \delta} M_{f,a,\beta}(z) \leq a(\delta) \ \forall \delta > 0\biggr\}.
\]
\citet{BSY2017} show that all Gaussian and multivariate-$t$ densities (amongst others) belong to $\mathcal{F}_{d,\theta}$ for appropriate $\theta \in \Theta$.  

Now, for $d_X, d_Y \in \mathbb{N}$ and $\vartheta=(\theta,\theta_Y) \in \Theta^2$, define
\[
	\mathcal{F}_{d_X,d_Y,\vartheta} := \Bigl\{ f \in \mathcal{F}_{d_X+d_Y,\theta} : f_Y  \in \mathcal{F}_{d_Y, \theta_Y}, f_Xf_Y \in \mathcal{F}_{d_X+d_Y,\theta} \Bigr\}
\]
and, for $b \geq 0$, let
\[
	\mathcal{F}_{d_X,d_Y,\vartheta}(b) := \Bigl\{ f \in \mathcal{F}_{d_X,d_Y,\vartheta} : I(f)  > b \Bigr\}.
\]
Thus, $\mathcal{F}_{d_X,d_Y,\vartheta}(b)$ consists of joint densities whose mutual information is greater than $b$.  In Theorem~\ref{Thm:Power} below, we will show that for a suitable choice of $b=b_n$ and for certain $\vartheta \in \Theta^2$, the power of the test defined in Lemma~\ref{Lemma:Size} converges to 1, uniformly over $\mathcal{F}_{d_X,d_Y,\vartheta}(b)$.

Before we can state this result, however, we must define the allowable choices of $k$ and the weight vectors.  Given $d \in \mathbb{N}$ and $\theta=(\alpha,\beta,\gamma,\nu,a) \in \Theta$ let
\[
	\tau_1(d,\theta) :=  \min \biggl\{ \frac{2 \alpha}{5\alpha+3d} \, , \, \frac{\alpha-d}{2\alpha} \, , \, \frac{4(\beta \wedge 1)}{4(\beta \wedge 1)+3d}\biggr\} 
\]
and
\[
	\tau_2(d,\theta) := \min \biggl\{1-\frac{d}{2 \beta} \, , \, 1- \frac{d/4}{\lfloor d/4 \rfloor +1} \biggr\}
\]
Note that $\min_{i=1,2} \tau_i(d,\theta) >0$ if and only if both $\alpha>d$ and $\beta > d/2$.  Finally, for $k \in \mathbb{N}$, let
\begin{align*}
	\mathcal{W}^{(k)} := \biggl\{ w = (w_1,\ldots,w_k)^T \in \mathbb{R}^k : &\sum_{j=1}^k w_j \frac{\Gamma(j+2 \ell/d)}{\Gamma(j)} =0 \quad \text{for} \, \, \ell=1, \ldots, \lfloor d/4 \rfloor \nonumber \\
	&\sum_{j=1}^k w_j= 1\, \text{and} \, \, w_j=0 \, \, \text{if}\, j \notin \{ \lfloor k/d \rfloor, \lfloor 2k/d \rfloor, \ldots, k\} \biggr\}.
\end{align*}
Thus, our weights sum to 1; the other constraints ensure that the dominant contributions to the bias of the unweighted Kozachenko--Leonenko estimator cancel out to sufficiently high order, and that the corresponding $j$th nearest neighbour distances are not too highly correlated.
\begin{thm}
\label{Thm:Power}
Fix $d_X,d_Y \in \mathbb{N}$, set $d=d_X+d_Y$ and fix $\vartheta = (\theta, \theta_Y) \in \Theta^2$ with 
\[
	\min\Bigl\{\tau_1(d,\theta)\, , \, \tau_1(d_Y,\theta_Y)\, , \,\tau_2(d,\theta)\, , \,\tau_2(d_Y,\theta_Y)\Bigr\}  >0.
\]
Let $k_0^*=k_{0,n}^*, k_Y^*=k_{Y,n}^*$ and $k^*=k_n^*$ denote any deterministic sequences of positive integers with $k_0^* \leq \min\{k_Y^*, k^*\}$, with $k_0^* / \log^5 n \rightarrow \infty$ and with
\[
	\max \biggl\{ \frac{k^*}{n^{\tau_1(d,\theta)-\epsilon}} \, , \, \frac{k_Y^*}{n^{\tau_1(d_Y,\theta_Y)-\epsilon}} \, , \, \frac{k^*}{n^{\tau_2(d,\theta)}} \, , \, \frac{k_{Y}^*}{n^{\tau_2(d_Y,\theta_Y)}} \biggr\} \rightarrow 0
\]
for some $\epsilon>0$. Also suppose that $w_Y=w_Y^{(k_Y)} \in \mathcal{W}^{(k_Y)}$ and $w=w^{(k)} \in \mathcal{W}^{(k)}$, and that $\limsup_n \max(\|w\|,\|w_Y\|) < \infty$. Then there exists a sequence $(b_n)$ such that $b_n=o(n^{-1/2})$ and with the property that for each $q \in (0,1)$ and any sequence $(B_n^*)$ with $B_n^* \rightarrow \infty$,
\[
	\inf_{B_n \geq B_n^*} \inf_{\substack{k_Y \in \{k_0^*, \ldots, k_Y^*\} \\ k \in \{k_0^*, \ldots, k^*\}}}\inf_{f \in \mathcal{F}_{d_X,d_Y,\vartheta}(b_n)} \mathbb{P}_f( \hat{I}_n > \hat{C}_q^{(n),B_n}) \rightarrow 1.
\]
\end{thm}
Theorem~\ref{Thm:Power} provides a strong guarantee on the ability of \texttt{MINTknown} to detect alternatives, uniformly over classes whose mutual information is at least $b_n$, where we may even have $b_n = o(n^{-1/2})$.  

\section{The case of unknown marginal distributions}
\label{Sec:UnknownMarginals}

We now consider the setting in which the marginal distributions of both $X$ and $Y$ are unknown.  Our test statistic remains the same, but now we estimate the critical value by permuting our sample in an attempt to mimic the behaviour of the test statistic under $H_0$.  More explicity, for some $B \in \mathbb{N}$, we propose independently of $(X_1,Y_1),\ldots,(X_n,Y_n)$ to simulate independent random variables $\tau_1, \ldots, \tau_B$ uniformly from $S_n$, the permutation group of $\{1, \ldots, n\}$, and for $b=1,\ldots,B$, set $Z_i^{(b)}:=(X_i, Y_{\tau_b(i)})$ and $\tilde{I}_n^{(b)} := \hat{I}_n(Z_1^{(b)}, \ldots, Z_n^{(b)})$.  For $q \in (0,1)$, we can now estimate $C_q^{(n)}$ by 
\[
	\tilde{C}_q^{(n),B}:= \inf \biggl\{ r \in \mathbb{R} : \frac{1}{B+1}\sum_{b=0}^B \mathbbm{1}_{\{\tilde{I}_n^{(b)} \geq r\}} \leq q\biggr\},
\]
where $\tilde{I}_n^{(0)} := \hat{I}_n$, and refer to the test that rejects $H_0$ if and only if $\hat{I}_n > \tilde{C}_q^{(n),B}$ as \texttt{MINTunknown}$(q)$.
\begin{lemma}
\label{Thm:size}
For any $q \in (0,1)$ and $B \in \mathbb{N}$, the \emph{\texttt{MINTunknown}}$(q)$ test has size at most $q$:
\[
\sup_{k \in \{1,\ldots,n-1\}} \sup_{(X,Y):I(X;Y) = 0} \mathbb{P}\bigl(\hat{I}_n > \tilde{C}_q^{(n),B}\bigr) \leq q.
\]
\end{lemma}
Note that $\hat{I}_n > \tilde{C}_q^{(n),B}$ if and only if
\begin{equation}
\label{Eq:Reject}
	\frac{1}{B+1}\sum_{b=0}^B \mathbbm{1}_{\{\tilde{I}_n^{(b)} \geq \hat{I}_n\}} \leq q.
\end{equation}
This shows that estimating either of the marginal entropies is unnecessary to carry out the test, since $\tilde{I}_n^{(b)}-\hat{I}_n=\hat{H}_n^{Z}-\tilde{H}_n^{(b)}$, where $\tilde{H}_n^{(b)} := \hat{H}_{n,k}^{d,w^Z}(Z_1^{(b)},\ldots,Z_n^{(b)})$ is the weighted Kozachenko--Leonenko joint entropy estimator based on the permuted data.


We now study the power of \texttt{MINTunknown}, and begin by introducing the classes of marginal densities that we consider. To define an appropriate notion of smoothness, for $z \in \{w:g(w) > 0\} =: \mathcal{W}$, $g \in \mathcal{F}_d$ and $\delta > 0$, let 
\begin{equation}
\label{Eq:r}
r_{z,g,\delta} := \biggl\{\frac{\delta e^{\Psi(k)}}{V_d(n-1)g(z)}\biggr\}^{1/d}.
\end{equation}
Now, for $A$ belonging to the class of Borel subsets of $\mathcal{W}$, denoted $\mathcal{B}(\mathcal{W})$, define 
\[
	M_{g}(A):= \sup_{\delta \in (0,2]} \sup_{z \in A} \biggl| \frac{1}{V_d r_{z,g,\delta}^d g(z)} \int_{B_z(r_{z,g,\delta})} g \,d \lambda_d  - 1\biggr|.
\]
Both $r_{z,g,\delta}$ and $M_{g}(\cdot)$ depend on $n$ and $k$, but for simplicity we suppress this in our notation.  Let $\phi=(\alpha,\mu,\nu,(c_n),(p_n)) \in (0,\infty)^3 \times [0,\infty)^\mathbb{N} \times [0,\infty)^\mathbb{N}=:\Phi$ and define
\begin{align*}
	\mathcal{G}_{d_X,d_Y,\phi} := \biggl\{ f &\in \mathcal{F}_{d_X+d_Y} : \max\{\mu_\alpha(f_X),\mu_\alpha(f_Y)\} \leq \mu, \max\{\|f_X\|_\infty, \|f_Y\|_\infty \} \leq \nu, \\
	 &\exists \mathcal{W}_n \in \mathcal{B}(\mathcal{X} \times \mathcal{Y}) \, \text{s.t.}\, M_{f_Xf_Y}(\mathcal{W}_n) \leq c_n, \int_{\mathcal{W}_n^c} f_Xf_Y \, d\lambda_d \leq p_n \ \forall n \in \mathbb{N}\biggr\}.
\end{align*}
In addition to controlling the $\alpha$th moment and uniform norms of the marginals $f_X$ and $f_Y$, the class $\mathcal{G}_{d,\phi}$ asks for there to be a (large) set $\mathcal{W}_n$ on which this product of marginal densities is uniformly well approximated by a constant over small balls.  This latter condition is satisfied by products of many standard parametric families of marginal densities, including normal, Weibull, Gumbel, logistic, gamma, beta, and $t$ densities, and is what ensures that nearest neighbour methods are effective in this context.

The corresponding class of joint densities we consider, for $\phi=(\alpha,\mu,\nu,(c_n),(p_n)) \in \Phi$, is
\begin{align*}
	\mathcal{H}_{d,\phi} := \biggl\{ f \in \mathcal{F}_{d}&: \mu_\alpha(f) \leq \mu, \|f\|_\infty \leq \nu, \\
	 &\exists \mathcal{Z}_n \in \mathcal{B}(\mathcal{Z}) \, \text{s.t.}\, M_{f}(\mathcal{Z}_n) \leq c_n, \int_{\mathcal{Z}_n^c} f \, d\lambda_d \leq p_n \ \forall n \in \mathbb{N}\biggr\}.
\end{align*}
In many cases, we may take $\mathcal{Z}_n = \{ z : f(z) \geq \delta_n \}$, for some appropriately chosen sequence $(\delta_n)$ with $\delta_n \rightarrow 0$ as $n \rightarrow \infty$.  For instance, suppose we fix $d \in \mathbb{N}$ and $\theta = (\alpha, \beta, \nu, \gamma, a) \in \Theta$.  Then, by \citet[][Lemma~12]{BSY2017}, there exists such a sequence $(\delta_n)$, as well as sequences $(c_n)$ and $(p_n)$, where 
\[
	\delta_n = \frac{k a(k/(n-1))^\frac{d}{\beta \wedge 1}}{n-1} \log(n-1), \quad c_n = \frac{15}{7} \frac{2^\frac{\beta \wedge 1}{d} d^{3/2}}{d+(\beta \wedge 1)} \log^{-\frac{\beta \wedge 1}{d}} (n-1),
\]
for large $n$ and $p_n = o( (k/n)^{\frac{\alpha}{\alpha+d} - \epsilon})$ for every $\epsilon > 0$, such that $\mathcal{F}_{d, \theta} \subseteq \mathcal{H}_{d, \phi}$ with $\phi=(\alpha, \mu, \nu, (c_n), (p_n)) \in \Phi$.
We are now in a position to state our main result on the power of \texttt{MINTunknown}.
\begin{thm}
\label{Thm:MINTunknownpower}
Let $d_X, d_Y \in \mathbb{N}$, let $d = d_X+d_Y$ and fix $\phi=(\alpha,\mu,\nu,(c_n),(p_n)) \in \Phi$ with $c_n \rightarrow 0$ and $p_n = o(1/\log n)$ as $n \rightarrow \infty$.  Let $k_0^* = k_{0,n}^*$ and $k_1^* = k_{1,n}^*$ denote two deterministic sequences of positive integers satisfying $k_0^* \leq k_1^*$, $k_0^*/ \log^2 n \rightarrow \infty$ and $(k_1^* \log^2 n)/n \rightarrow 0$.  Then for any $b > 0$, $q \in (0,1)$ and any sequence $(B_n^*)$ with $B_n^* \rightarrow \infty$ as $n \rightarrow \infty$, 
\[
	\inf_{B_n \geq B_n^*} \inf_{k \in \{k_0^*,\ldots,k_1^*\}} \inf_{\substack{f \in \mathcal{G}_{d_X,d_Y,\phi} \cap \mathcal{H}_{d,\phi}:\\ I(f) \geq b}} \mathbb{P}_f( \hat{I}_n > \tilde{C}_q^{(n),B_n}) \rightarrow 1
\]
as $n \rightarrow \infty$.
\end{thm}
Theorem~\ref{Thm:MINTunknownpower} shows that \texttt{MINTunknown} is uniformly consistent against a wide class of alternatives.

\section{Regression setting}
\label{Sec:Regression}

In this section we aim to extend the ideas developed above to the problem of goodness-of-fit testing in linear models.  Suppose we have independent and identically distributed pairs $(X_1,Y_1),\ldots,(X_n,Y_n)$ taking values in $\mathbb{R}^p \times \mathbb{R}$, with $\mathbb{E}(Y_1^2) < \infty$ and $\mathbb{E}(X_1X_1^T)$ finite and positive definite.  Then 
\[
\beta_0 := \argmin_{\beta \in \mathbb{R}^p} \mathbb{E}\{(Y_1-X_1^T\beta)^2\}
\]
is well-defined, and we can further define $\epsilon_i := Y_i - X_i^T\beta_0$ for $i=1,\ldots,n$.  We show in the proof of Theorem~\ref{Thm:MINTregressionpower} below that $\mathbb{E}(\epsilon_1X_1) = 0$, but for the purposes of interpretability and inference, it is often convenient if the random design linear model
\[
Y_i = X_i^T\beta_0 + \epsilon_i, \quad i=1,\ldots,n,
\]
holds with $X_i$ and $\epsilon_i$ independent.  A goodness-of-fit test of this property amounts to a test of $H_0: X_1 \perp\!\!\!\perp \epsilon_1$ against $H_1: X_1 \not\!\perp\!\!\!\perp \epsilon_1$.  The main difficulty here is that $\epsilon_1,\ldots,\epsilon_n$ are not observed directly.  Given an estimator $\hat{\beta}$ of $\beta_0$, the standard approach for dealing with this problem is to compute residuals $\hat{\epsilon}_i := Y_i - X_i^T\hat{\beta}$ for $i=1,\ldots,n$, and use these as a proxy for $\epsilon_1,\ldots,\epsilon_n$.  Many introductory statistics textbooks, e.g.\ \citet[][Section~2.3.4]{Dobson2002}, \citet[][Section~5.2]{Dalgaard2002} suggest examining for patterns plots of residuals against fitted values, as well as plots of residuals against each covariate in turn, as a diagnostic, though it is difficult to formalise this procedure.  (It is also interesting to note that when applying the \texttt{plot} function in \texttt{R} to an object of type \texttt{lm}, these latter plots of residuals against each covariate in turn are not produced, presumably because it would be prohibitively time-consuming to check them all in the case of many covariates.) 

The naive approach based on our work so far is simply to use the permutation test of Section~\ref{Sec:UnknownMarginals} on the data $(X_1, \hat{\epsilon}_1), \ldots, (X_n, \hat{\epsilon}_n)$. Unfortunately, calculating the test statistic $\hat{I}_n$ on permuted data sets does not result in an exchangeable sequence, which makes it difficult to ensure that this test has the nominal size $q$.  
To circumvent this issue, we assume that the marginal distribution of $\epsilon_1$ under $H_0$ has $\mathbb{E}(\epsilon_1) = 0$ and is known up to a scale factor $\sigma^2 := \mathbb{E}(\epsilon_1^2)$; in practice, it will often be the case that this marginal distribution is taken to be $N(0,\sigma^2)$, for some unknown $\sigma > 0$.  We also assume that we can sample from the density $f_\eta$ of $\eta_1 := \epsilon_1/\sigma$; of course, this is straightforward in the normal distribution case above.  Let $X=( X_{1} \cdots X_{n})^T$, $Y=(Y_{1}, \ldots, Y_{n})^T$, and suppose the vector of residuals $\hat{\epsilon}  := (\hat{\epsilon}_1,\ldots,\hat{\epsilon}_n)^T$ is computed from the least squares estimator $\hat{\beta} := (X^TX)^{-1}X^TY$.  We then define standardised residuals by $\hat{\eta}_i := \hat{\epsilon}_i / \hat{\sigma}$, for $i=1,\ldots,n$, where $\hat{\sigma}^2:=n^{-1}\|\hat{\epsilon}\|^2$; these standardised residuals are invariant under changes of scale of $\epsilon := (\epsilon_1,\ldots,\epsilon_n)^T$.  Suppressing the dependence of our entropy estimators on $k$ and the weights $w_1,\ldots,w_k$ for notational simplicity, our test statistic is now given by
\[
	\check{I}_n^{(0)}:= \hat{H}_{n}^{p}(X_1, \ldots, X_n) + \hat{H}_{n}^{1}(\hat{\eta}_1, \ldots, \hat{\eta}_n) - \hat{H}_n^{p+1}\bigl( (X_1, \hat{\eta}_1), \ldots, (X_n, \hat{\eta}_n) \bigr).
\]
Writing $\eta := \epsilon/\sigma$, we note that
\[
	\hat{\eta}_i = \frac{1}{\hat{\sigma}}(Y_i - X_i^T\hat{\beta}) = \frac{1}{\hat{\sigma}}\bigl\{\epsilon_i - X_i^T(\hat{\beta} - \beta_0)\bigr\} = \frac{n^{1/2}\bigl\{\eta_i - X_i^T (X^TX)^{-1} X^T \eta\bigr\}}{\|\eta - X^T (X^TX)^{-1} X^T \eta\|},
\]
whose distribution does not depend on the unknown $\beta_0$ or $\sigma^2$.  Let $\{\eta^{(b)}=(\eta_1^{(b)}, \ldots, \eta_n^{(b)})^T:b=1,\ldots,B\}$ denote independent random vectors, whose components are generated independently from $f_\eta$. For $b=1, \ldots, B$ we then set $\hat{s}^{(b)} := n^{-1/2}\|\hat{\eta}^{(b)}\|$ and, for $i=1,\ldots,n$, let 
\[
	\hat{\eta}_i^{(b)} := \frac{1}{\hat{s}^{(b)}}\bigl\{\eta_i^{(b)} - X_i^T(X^TX)^{-1} X^T \eta^{(b)}\bigr\}.
\]
We finally compute
\[
	\check{I}_n^{(b)}:= \hat{H}_n^p(X_1, \ldots, X_n) + \hat{H}_n^1(\hat{\eta}_1^{(b)}, \ldots, \hat{\eta}_n^{(b)}) - \hat{H}_n^{p+1} \bigl( (X_1, \hat{\eta}_1^{(b)}), \ldots, (X_n, \hat{\eta}_n^{(b)}) \bigr).
\]
Analogously to our development in Sections~\ref{Sec:KnownMarginals} and~\ref{Sec:UnknownMarginals}, we can then define a critical value by
\[
\check{C}_q^{(n),B} := \inf\biggl\{r \in \mathbb{R}: \frac{1}{B+1}\sum_{b=0}^B \mathbbm{1}_{\{\check{I}_n^{(b)} \geq r\}} \leq q\biggr\}. 
\]
Conditional on $X$ and under $H_0$, the sequence $(\check{I}_n^{(0)}, \ldots, \check{I}_n^{(B)})$ is independent and identically distributed and unconditionally it is an exchangeable sequence under $H_0$.  This is the crucial observation from which we can immediately derive that the resulting test has the nominal size. 
\begin{lemma}
\label{Lemma:MINTregressionsize}
For each $q \in (0,1)$ and $B \in \mathbb{N}$, the \emph{\texttt{MINTregression}($q$)} test that rejects $H_0$ if and only if $\check{I}_n^{(0)} > \check{C}_q^{(n),B}$ has size at most $q$:
\[
\sup_{k \in \{1,\ldots,n-1\}} \sup_{(X,Y):I(X;Y) = 0} \mathbb{P}\bigl(\check{I}_n > \check{C}_q^{(n),B}\bigr) \leq q.
\]
\end{lemma}
As in previous sections, we are only interested in the differences $\check{I}_n^{(0)} - \check{I}_n^{(b)}$ for $b=1,\ldots,B$, and in these differences, the $\hat{H}_n^p(X_1, \ldots, X_n)$ terms cancel out, so these marginal entropy estimators need not be computed.


In fact, to simplify our power analysis, it is more convenient to define a slightly modified test, which also has the nominal size.  Specifically, we assume for simplicity that $m := n/2$ is an integer, and consider a test in which the sample is split in half, with the second half of the sample used to calculate the estimators $\hat{\beta}_{(2)}$ and $\hat{\sigma}_{(2)}^2$ of $\beta_0$ and $\sigma^2$ respectively.   On the first half of the sample, we calculate
\[
	\hat{\eta}_{i,(1)} := \frac{Y_i - X_i^T\hat{\beta}_{(2)}}{\hat{\sigma}_{(2)}}
\]
for $i=1, \ldots, m$ and the test statistic
\[
	\breve{I}_n^{(0)} := \hat{H}_m^p(X_1, \ldots, X_m) + \hat{H}_m^1(\hat{\eta}_{1,(1)}, \ldots, \hat{\eta}_{m,(1)}) - \hat{H}_m^{p+1}\bigl( (X_1, \hat{\eta}_{1,(1)}), \ldots, (X_m, \hat{\eta}_{m,(1)}) \bigr).
\]
Corresponding estimators $\{\breve{I}_n^{(b)}:b=1,\ldots,B\}$ based on the simulated data may also be computed using the same sample-splitting procedure, and we then obtain the critical value $\breve{C}_q^{(n),B}$ in the same way as above.  The advantage from a theoretical perspective of this approach is that, conditional on $\hat{\beta}_{(2)}$ and $\hat{\sigma}_{(2)}^2$, the random variables $\hat{\eta}_{1,(1)},\ldots,\hat{\eta}_{m,(1)}$ are independent and identically distributed.

To describe the power properties of \texttt{MINTregression}, we first define several densities: for $\gamma \in \mathbb{R}^p$ and $s > 0$, let $f_{\hat{\eta}}^{\gamma,s}$ and $f_{\hat{\eta}^{(1)}}^{\gamma,s}$ denote the densities of $\hat{\eta}_{1}^{\gamma,s} := (\eta_1 - X_1^T\gamma)/s$ and $\hat{\eta}_{1}^{(1),\gamma,s} := (\eta_1^{(1)} - X_1^T\gamma)/s$ respectively; further, let $f_{X,\hat{\eta}}^{\gamma,s}$ and $f_{X,\hat{\eta}^{(1)}}^{\gamma,s}$ be the densities of $(X_1,\hat{\eta}_{1}^{\gamma,s})$ and $(X_1,\hat{\eta}_{1}^{(1),\gamma,s})$ respectively.  Note that imposing assumptions on these densities amounts to imposing assumptions on the joint density $f$ of $(X,\epsilon)$.  For $\Omega := \Theta \times \Theta \times (0,\infty) \times (0,1) \times (0,\infty) \times (0,\infty)$, and $\omega = (\theta_1,\theta_2,r_0,s_0,\Lambda,\lambda_0)$, we therefore let $\mathcal{F}_{p+1,\omega}^*$ denote the class of joint densities $f$ of $(X,\epsilon)$ satisfying the following three requirements:
\begin{enumerate}[(i)]
\item
\[
\Bigl\{f_{\hat{\eta}}^{\gamma,s} : \gamma \in B_0(r_0), s \in [s_0,1/s_0] \Bigr\} \cup \Bigl\{f_{\hat{\eta}^{(1)}}^{\gamma,s} : \gamma \in B_0(r_0), s \in [s_0,1/s_0] \Bigr\} \subseteq \mathcal{F}_{1,\theta_1}
\]
and 
\[
	\Bigl\{f_{X,\hat{\eta}}^{\gamma,s} : \gamma \in B_0(r_0), s \in [s_0,1/s_0] \Bigr\} \cup \Bigl\{f_{X,\hat{\eta}^{(1)}}^{\gamma,s} : \gamma \in B_0(r_0), s \in [s_0,1/s_0] \Bigr\} \subseteq \mathcal{F}_{p+1,\theta_2}.
\]
\item
\begin{equation}
\label{Eq:ii,1}
\sup_{\gamma \in B_0(r_0)} \max \Bigl\{\mathbb{E} \log^2 f_{\hat{\eta}}^{\gamma,1}(\eta_1) \, , \, \mathbb{E} \log^2 f_{\hat{\eta}^{(1)}}^{\gamma,1}(\eta_1)\Bigr\} \leq \Lambda, \quad 
\end{equation}
and
\begin{equation}
\label{Eq:ii,2}
\sup_{\gamma \in B_0(r_0)} \max \Bigl\{\mathbb{E} \log^2 f_{\eta}\bigl(\hat{\eta}_1^{\gamma,1}\bigr) \, , \, \mathbb{E} \log^2 f_{\eta}\bigl(\hat{\eta}_1^{(1),\gamma,1}\bigr)\Bigr\} \leq \Lambda.
\end{equation}
\item Writing $\Sigma := \mathbb{E}(X_1X_1^T)$, we have $\lambda_{\min}(\Sigma) \geq \lambda_0$.
\end{enumerate}
The first of these requirements ensures that we can estimate efficiently the marginal entropy of our scaled residuals, as well as the joint entropy of these scaled residuals and our covariates.  The second condition is a moment condition that allows us to control $|H(\eta_1 - X_1^T\gamma) - H(\eta_1)|$ (and similar quantities) in terms of $\|\gamma\|$, when $\gamma$ belongs to a small ball around the origin.  To illustrate the second part of this condition, it is satisfied, for instance, if $f_\eta$ is a standard normal density and $\mathbb{E}(\|X_1\|^4) < \infty$, or if $f_\eta$ is a $t$ density and $\mathbb{E}(\|X_1\|^\alpha) < \infty$ for some $\alpha > 0$; the first part of the condition is a little more complicated but similar.  The final condition is very natural for random design regression problems.


By the same observation on the sequence $(\breve{I}_n^{(0)},\ldots,\breve{I}_n^{(B)})$ as was made regarding the sequence $(\check{I}_n^{(0)},\ldots,\check{I}_n^{(B)})$ just before Lemma~\ref{Lemma:MINTregressionsize}, we see that the sample-splitting version of the \texttt{MINTregression}$(q)$ test has size at most $q$.  
\begin{thm}
\label{Thm:MINTregressionpower}
Fix $p \in \mathbb{N}$ and $\omega = (\theta_1,\theta_2,r_0,s_0,\Lambda,\lambda_0) \in \Omega$, where the first component of $\theta_2$ is $\alpha_2 \geq 4$ and the second component of $\theta_1$ is $\beta_1 \geq 1$.  Assume that 
\[
	\min\Bigl\{\tau_1(1,\theta_1)\, , \, \tau_1(p+1,\theta_2)\, , \,\tau_2(1,\theta_1)\, , \,\tau_2(p+1,\theta_2)\Bigr\}  >0.
\]
Let $k_0^*=k_{0,n}^*, k_\eta^*=k_{\eta,n}^*$ and $k^*=k_n^*$ denote any deterministic sequences of positive integers with $k_0^* \leq \min\{k_\eta^*, k^*\}$, with $k_0^* / \log^5 n \rightarrow \infty$ and with
\[
	\max \biggl\{ \frac{k^*}{n^{\tau_1(p+1,\theta_2)-\epsilon}} \, , \, \frac{k_\eta^*}{n^{\tau_1(1,\theta_1)-\epsilon}} \, , \, \frac{k^*}{n^{\tau_2(p+1,\theta_2)}} \, , \, \frac{k_{\eta}^*}{n^{\tau_2(1,\theta_1)}} \biggr\} \rightarrow 0
\]
for some $\epsilon>0$. Also suppose that $w_\eta=w_\eta^{(k_\eta)} \in \mathcal{W}^{(k_\eta)}$ and $w=w^{(k)} \in \mathcal{W}^{(k)}$, and that $\limsup_n \max(\|w\|,\|w_\eta\|) < \infty$.  Then for any sequence $(b_n)$ such that $n^{1/2}b_n \rightarrow \infty$, any $q \in (0,1)$ and any sequence $(B_n^*)$ with $B_n^* \rightarrow \infty$,
\[
	\inf_{B_n \geq B_n^*} \inf_{\substack{k_\eta \in \{k_0^*, \ldots, k_\eta^*\} \\ k \in \{k_0^*, \ldots, k^*\}}}\inf_{f \in \mathcal{F}_{1,p+1}^*:I(f) \geq b_n} \mathbb{P}_f( \breve{I}_n > \breve{C}_q^{(n),B_n}) \rightarrow 1.
\]
\end{thm}
Finally in this section, we consider partitioning our design matrix as $X = (X^* \, X^{**}) \in \mathbb{R}^{n \times (p_0 + p_1)}$, with $p_0 + p_1 = p$, and describe an extension of \texttt{MINTregression} to cases where we are interested in testing the independence between $\epsilon$ and $X^*$.  For instance, $X^{**}$ may consist of an intercept term, or transformations of variables in $X^*$, as in the real data example presented in Section~\ref{Sec:RealData} below.  Our method for simulating standardised residual vectors $\{\hat{\eta}^{(b)}:b = 1,\ldots,B\}$ remains unchanged, but our test statistic and corresponding null statistics become
\begin{align*}
\bar{I}_n^{(0)} &:= \hat{H}_n^p(X_1^*, \ldots, X_n^*) + \hat{H}_n^1(\hat{\eta}_1, \ldots, \hat{\eta}_n) - \hat{H}_n^{p+1} \bigl( (X_1^*, \hat{\eta}_1), \ldots, (X_n^*, \hat{\eta}_n) \bigr) \\
\bar{I}_n^{(b)} &:= \hat{H}_n^p(X_1^*, \ldots, X_n^*) + \hat{H}_n^1(\hat{\eta}_1^{(b)}, \ldots, \hat{\eta}_n^{(b)}) - \hat{H}_n^{p+1} \bigl( (X_1^*, \hat{\eta}_1^{(b)}), \ldots, (X_n^*, \hat{\eta}_n^{(b)}) \bigr), \ b=1,\ldots,B.
\end{align*}
The sequence $(\bar{I}_n^{(0)},\bar{I}_n^{(1)},\ldots,\bar{I}_n^{(B)})$ is again exchangeable under $H_0$, so a $p$-value for this modified test is given by 
\[
	\frac{1}{B+1}\sum_{b=0}^B \mathbbm{1}_{\{\bar{I}_n^{(b)} \geq \bar{I}_n^{(0)}\}}.
\]

\section{Numerical studies}
\label{Sec:Simulations}

\subsection{Practical considerations}

For practical implementation of the \texttt{MINTunknown} test, we consider both a direct, data-driven approach to choosing $k$, and a multiscale approach that averages over a range of values of $k$.  To describe the first method, let $\mathcal{K} \subseteq \{1,\ldots,n-1\}$ denote a plausible set of values of $k$.  For a given $N \in \mathbb{N}$ and independently of the data, generate $\tau_1, \ldots, \tau_{2N}$ independently and uniformly from $S_n$, and for each $k \in \mathcal{K}$ and $j=1,\ldots,2N$ let $\hat{H}_{n,k}^{(j)}$ be the (unweighted) Kozachenko--Leonenko joint entropy estimate with tuning parameter $k$ based on the sample $\{(X_i,Y_{\tau_j(i)}) : i=1,\ldots,n\}$. We may then choose
\[
	\hat{k} := \sargmin_{k \in \mathcal{K}} \sum_{j=1}^N (\hat{H}_{n,k}^{(2j)} - \hat{H}_{n,k}^{(2j-1)} )^2,
\]
where $\sargmin$ denotes the smallest index of the $\argmin$ in the case of a tie.  In our simulations below, which use $N = 100$, we refer to the resulting test as \texttt{MINTauto}. 

For our multiscale approach, we again let $\mathcal{K} \subseteq \{1, \ldots, n-1\}$ and, for $k \in \mathcal{K}$, let $\hat{h}^{(0)}(k) := \hat{H}_{n,k}$ denote the (unweighted) Kozachenko--Leonenko entropy estimate with tuning parameter $k$ based on the original data $(X_1,Y_1), \ldots, (X_n,Y_n)$.  Now, for $b=1, \ldots, B$ and $k \in \mathcal{K}$, we let $\hat{h}^{(b)}(k) := \tilde{H}_{n,(k)}^{(b)}$ denote the Kozachenko--Leonenko entropy estimate with tuning parameter $k$ based on the permuted data $Z_1^{(b)}, \ldots, Z_n^{(b)}$. Writing $\bar{h}^{(b)} := |\mathcal{K}|^{-1} \sum_{k \in \mathcal{K}} \hat{h}^{(b)}(k)$ for $b=0,1,\ldots,B$, we then define the $p$-value for our test to be
\[
	\frac{1}{B+1} \sum_{b=0}^B \mathbbm{1}_{\{ \bar{h}^{(0)} \geq \bar{h}^{(b)}\}}.
\]
By the exchangeability of $(\bar{h}^{(0)},\bar{h}^{(1)},\ldots,\bar{h}^{(B)})$ under $H_0$, the corresponding test has the nominal size.  We refer to this test as \texttt{MINTav}, and note that if $\mathcal{K}$ is taken to be a singleton set then we recover \texttt{MINTunknown}.  In our simulations below, we took $\mathcal{K} = \{1,\ldots,20\}$ and $B=100$.

 
\subsection{Simulated data}
\label{Sec:SimulatedData}
To study the empirical performance of our methods, we first compare our tests to existing approaches through their performance on simulated data.  For comparison, we present corresponding results for a test based on the empirical copula process decribed by \citet{Kojadinovic2009} and implemented in the \texttt{R} package \texttt{copula} \citep{HKMY2017}, a test based on the HSIC implemented in the \texttt{R} package \texttt{dHSIC} \citep{PfisterPeters2017}, a test based on the distance covariance implemented in the \texttt{R} package \texttt{energy} \citep{RizzoSzekely2017} and the improvement of Hoeffding's test described in \citet{BergsmaDassios2014} and implemented in the \texttt{R} package \texttt{TauStar} \citep{Weihs16}.  We also present results for an oracle version of our tests, denoted simply as \texttt{MINT}, which for each parameter value in each setting, uses the best (most powerful) choice of $k$.  Throughout, we took $q=0.05$ and $n=200$, ran 5000 repetitions for each parameter setting, and for our comparison methods, used the default tuning parameter values recommended by the corresponding authors.  We consider three classes of data generating mechanisms, designed to illustrate different possible types of dependence:
\begin{enumerate}[(i)]
\item For $l \in \mathbb{N}$ and $(x,y) \in [-\pi,\pi]^2$, define the density function
\[
	f_l(x,y) = \frac{1}{\pi^2}\{1+ \sin(lx)\sin(ly)\}.
\]
This class of densities, which we refer to as sinusoidal, are identified by \citet{Sejdinovic13} as challenging ones to detect dependence, because as $l$ increases, the dependence becomes increasingly localised, while the marginal densities are uniform on $[-\pi,\pi]$ for each $l$. Despite this, by the periodicity of the sine function, we have that the mutual information does not depend on $l$: indeed,
\begin{align*}
	I(f_l) &= \frac{1}{\pi^2} \int_{-\pi}^{\pi} \int_{-\pi}^{\pi} \{1+ \sin(lx)\sin(ly)\} \log \bigl( 1+ \sin(lx)\sin(ly) \bigr) \,dx \,dy \\
	&= \frac{1}{\pi^2} \int_{-\pi}^{\pi} \int_{-\pi}^{\pi} (1+ \sin u\sin v) \log ( 1+ \sin u\sin v ) \,du \,dv \approx 0.0145.
\end{align*}
\item Let $L, \Theta, \epsilon_1, \epsilon_2$ be independent with $L \sim U(\{1, \ldots, l\})$ for some parameter $l \in \mathbb{N}$, $\Theta \sim U[0, 2 \pi]$, and $\epsilon_1, \epsilon_2 \sim N(0,1)$.  Set $X=L \cos \Theta + \epsilon_1/4$ and $Y= L \sin \Theta + \epsilon_2/4$.  For large values of $l$, the distribution of $(X/l,Y/l)$ approaches the uniform distribution on the unit disc.
\item Let $X, \epsilon$ be independent with $X \sim U[-1,1]$, $\epsilon \sim N(0,1)$, and for a parameter $\rho \in [0,\infty)$, let $Y= |X|^\rho \epsilon$.
\end{enumerate}
For each of these three classes of data generating mechanisms, we also consider a corresponding multivariate setting in which we wish to test the independence of $X$ and $Y$ when $X=(X_1,X_2)^T, Y=(Y_1,Y_2)^T$.  Here, $(X_1,Y_1), X_2, Y_2$ are independent, with $X_1$ and $Y_1$ having the dependence structures described above, and $X_2,Y_2 \sim U(0,1)$.  In these multivariate settings, the generalisations of the \texttt{TauStar} test were too computationally intensive, so were omitted from the comparison. 


The results are presented in Figure~\ref{Fig:Sinus}.  Unsurprisingly, there is no uniformly most powerful test among those considered, and if the form of the dependence were known in advance, it may well be possible to design a tailor-made test with good power.  Nevertheless, Figure~\ref{Fig:Sinus} shows that, especially in the first and second of the three settings described above, the \texttt{MINT} and \texttt{MINTav} approaches have very strong performance.  In these examples, the dependence becomes increasingly localised as $l$ increases, and the flexibility to choose a smaller value of $k$ in such settings means that \texttt{MINT} approaches are particularly effective.  Where the dependence is more global in nature, such as in setting~(iii), other approaches may be better suited, though even here, \texttt{MINT} is competitive.  Interestingly, the multiscale \texttt{MINTav} appears to be much more effective than the \texttt{MINTauto} approach of choosing a single value of $k$; indeed, in setting~(iii),~\texttt{MINTav} even outperforms the oracle choice of $k$. 


\begin{figure}%
\centering
\subfigure{\label{Fig:Sinus}\includegraphics[width=0.28\linewidth]{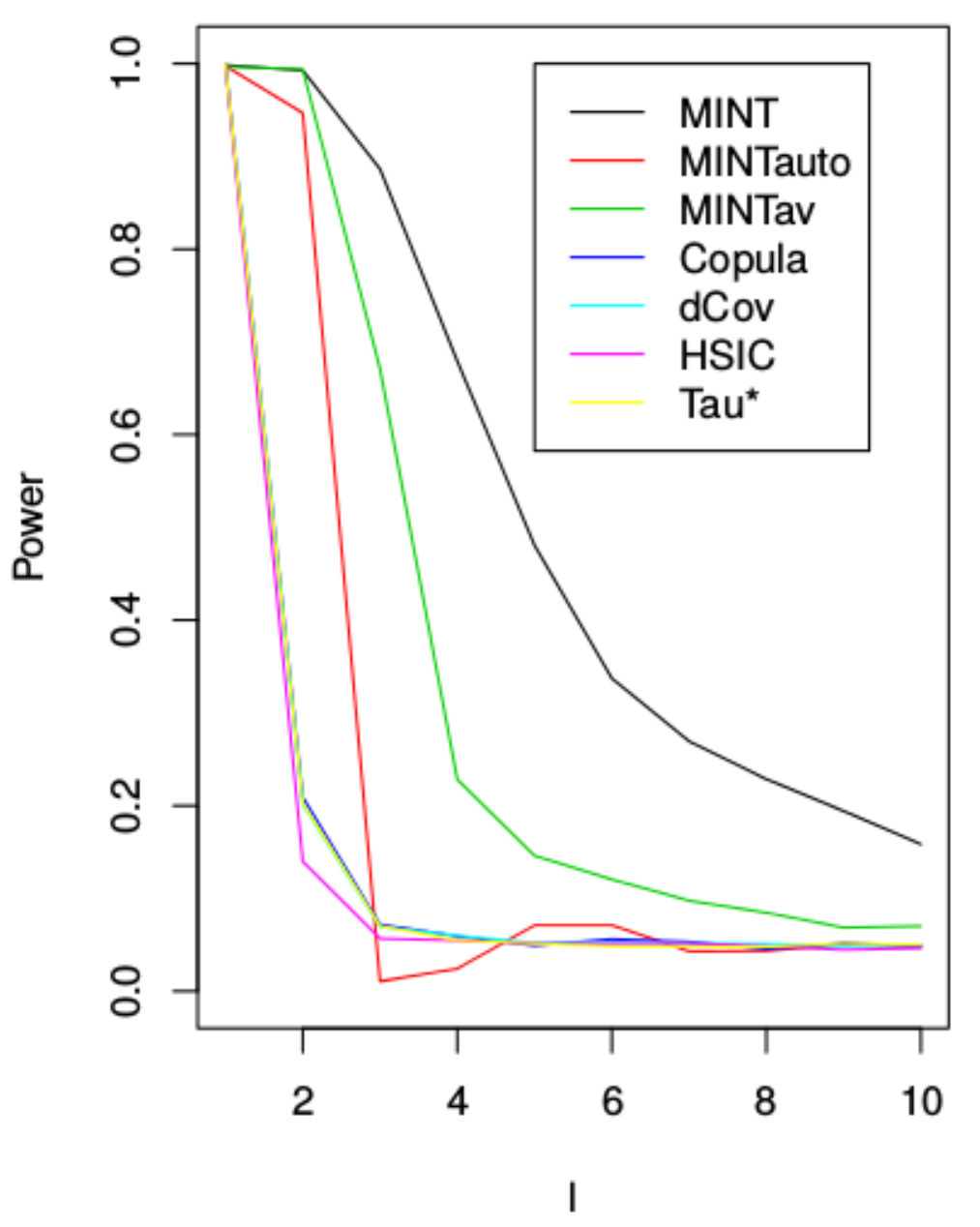}}\qquad
\subfigure{\label{Fig:Circular}\includegraphics[width=0.28\linewidth]{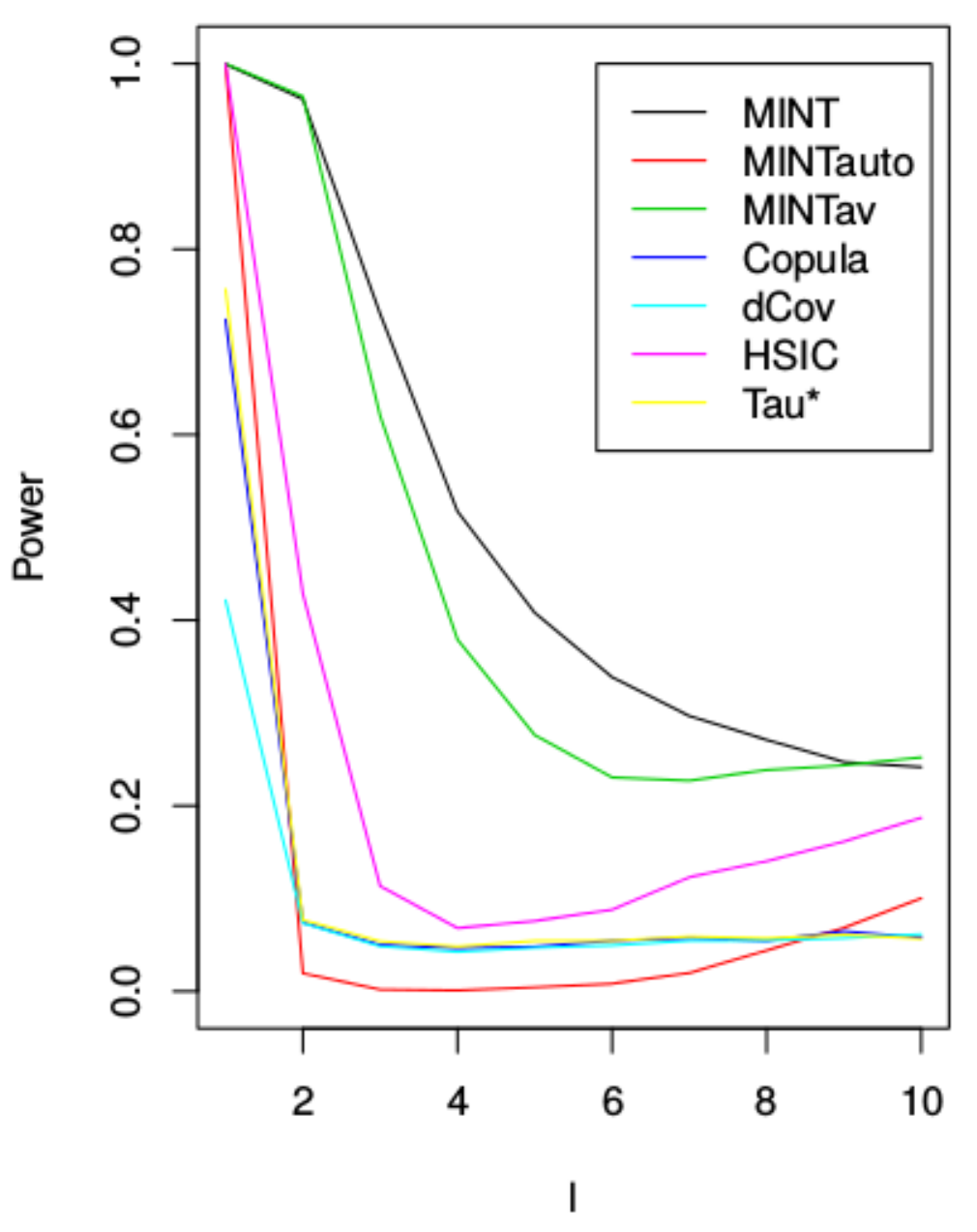}} \qquad
\subfigure{\label{Fig:Multiplicative}\includegraphics[width=0.28\linewidth]{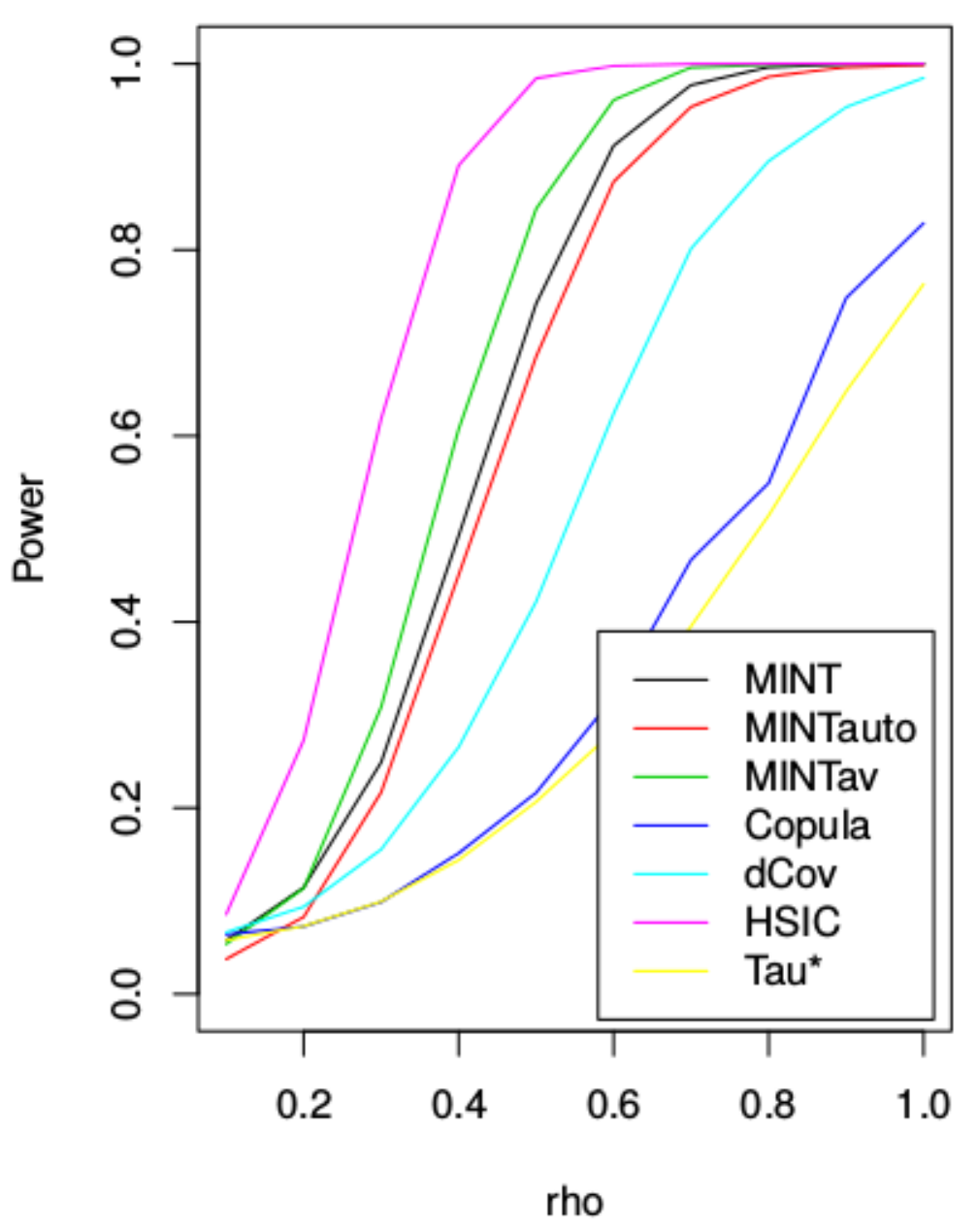}} \\
\subfigure{\label{Fig:MultiSinus}\includegraphics[width=0.28\linewidth]{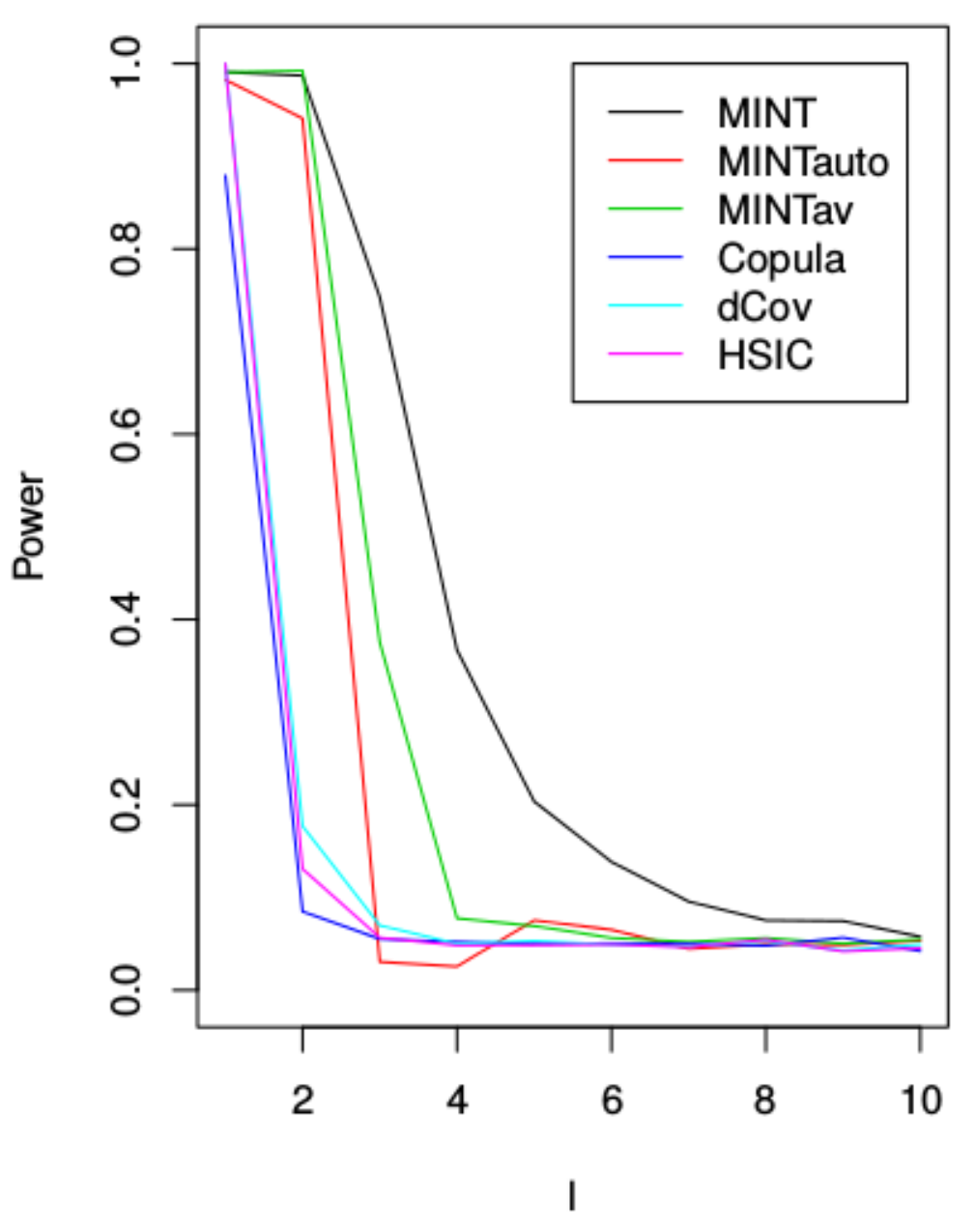}}\qquad
\subfigure{\label{Fig:MultiCircular}\includegraphics[width=0.28\linewidth]{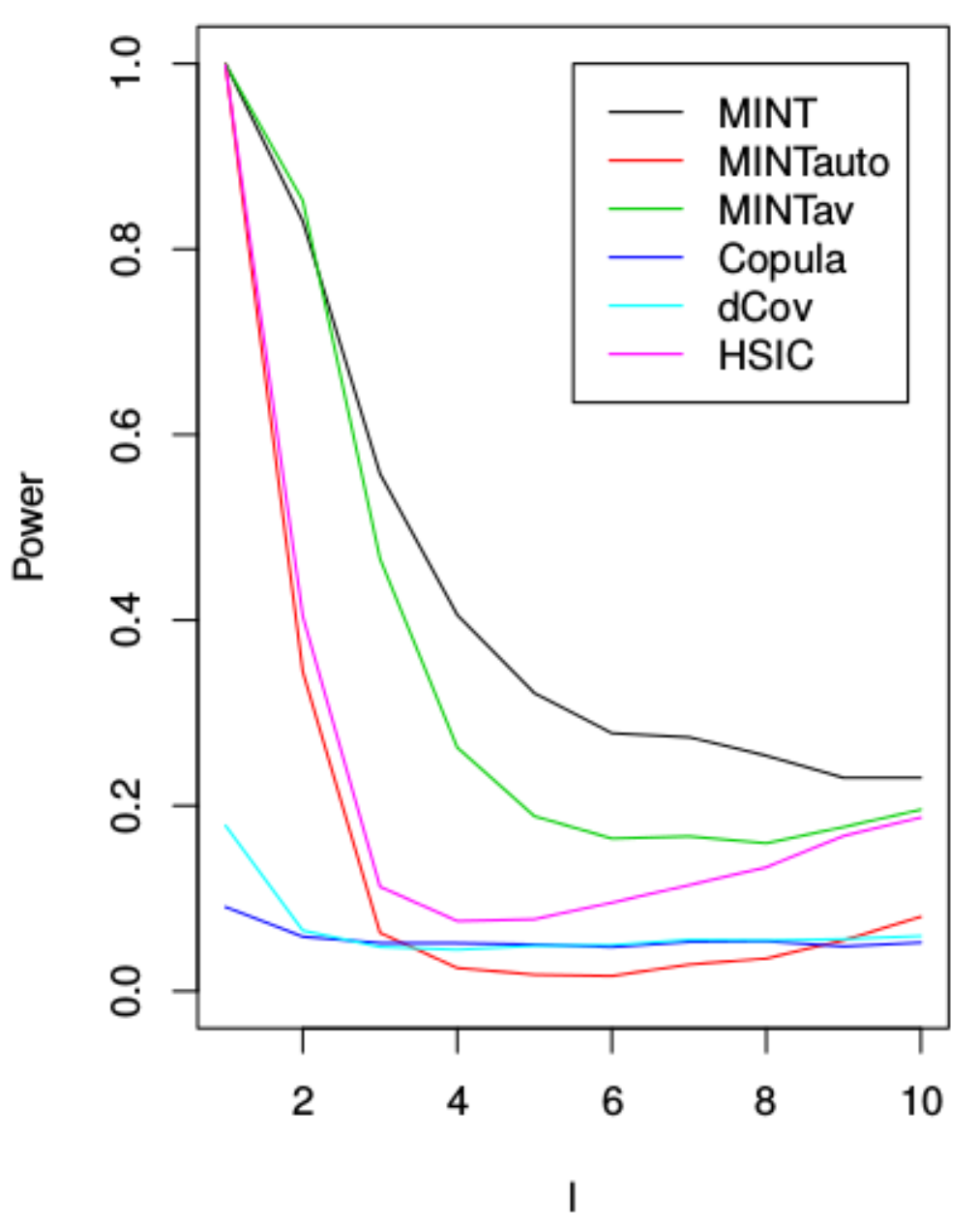}} \qquad
\subfigure{\label{Fig:MultiMultiplicative}\includegraphics[width=0.28\linewidth]{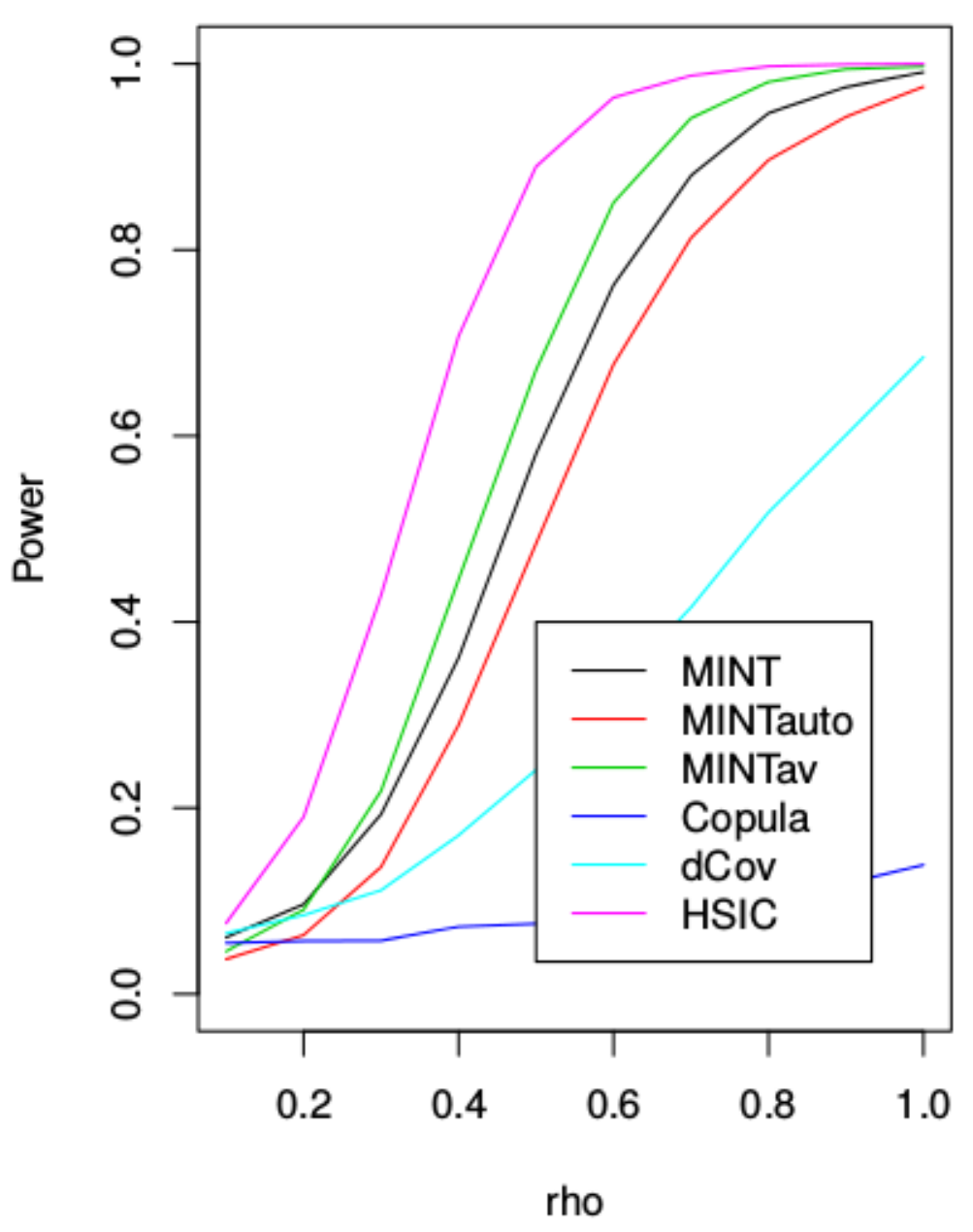}}
\caption{Power curves for the different tests considered, for settings~(i) (left), (ii) (middle) and~(iii) (right).  The marginals are univariate (top) and bivariate (bottom).}
\label{Fig:Simulations}
\end{figure}

\subsection{Real data}
\label{Sec:RealData}

In this section we illustrate the use of \texttt{MINTregression} on a data set comprising the average minimum daily January temperature, between 1931 and 1960 and measured in degrees Fahrenheit, in 56 US cities, indexed by their latitude and longitude \citep{Peixoto1990}. Fitting a normal linear model to the data with temperature as the response and latitude and longitude as covariates leads to the diagnostic plots shown in Figures~\ref{Fig:Diagnostic1} and~\ref{Fig:Diagnostic2}.
Based on these plots, it is relatively difficult to assess the strength of evidence against the model. Centering all of the variables and running \texttt{MINTregression} with $k_\eta=6$, $k=3$ and $B=1785 = \lfloor 10^5/56 \rfloor$ yielded a $p$-value of 0.00224, providing strong evidence that the normal linear model is not a good fit to the data.  Further investigation is possible via partial regression: in Figure~\ref{Fig:PartialRegression} we plot residuals based on a simple linear fit of Temperature on Latitude against the residuals of another linear fit of Longitude on Latitude.  This partial regression plot is indicative of a cubic relationship between Temperature and Longitude after removing the effects of Latitude.
\begin{figure}%
\centering
\subfigure{\label{Fig:Diagnostic1}\includegraphics[width=0.25\linewidth]{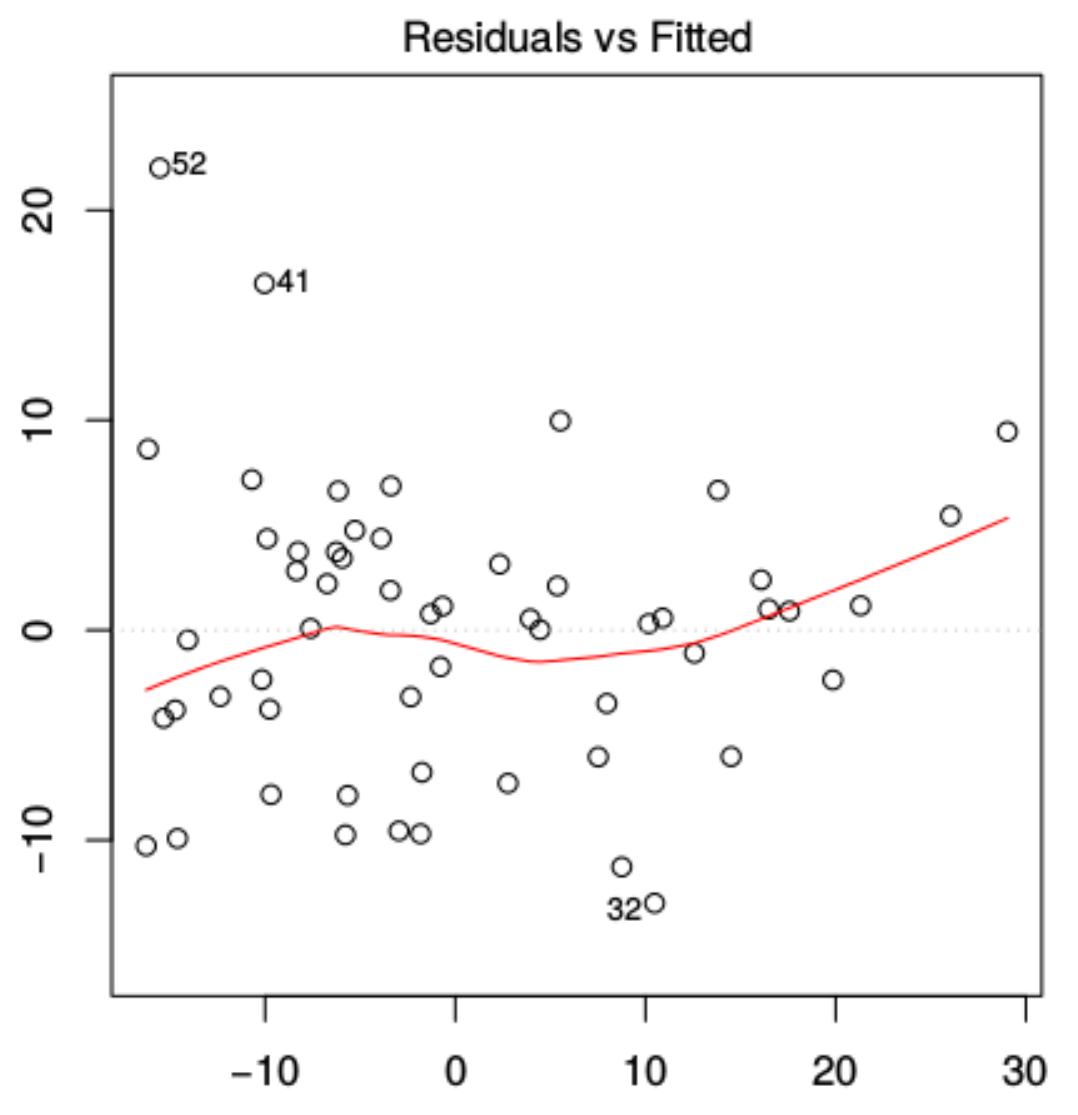}}\qquad
\subfigure{\label{Fig:Diagnostic2}\includegraphics[width=0.25\linewidth]{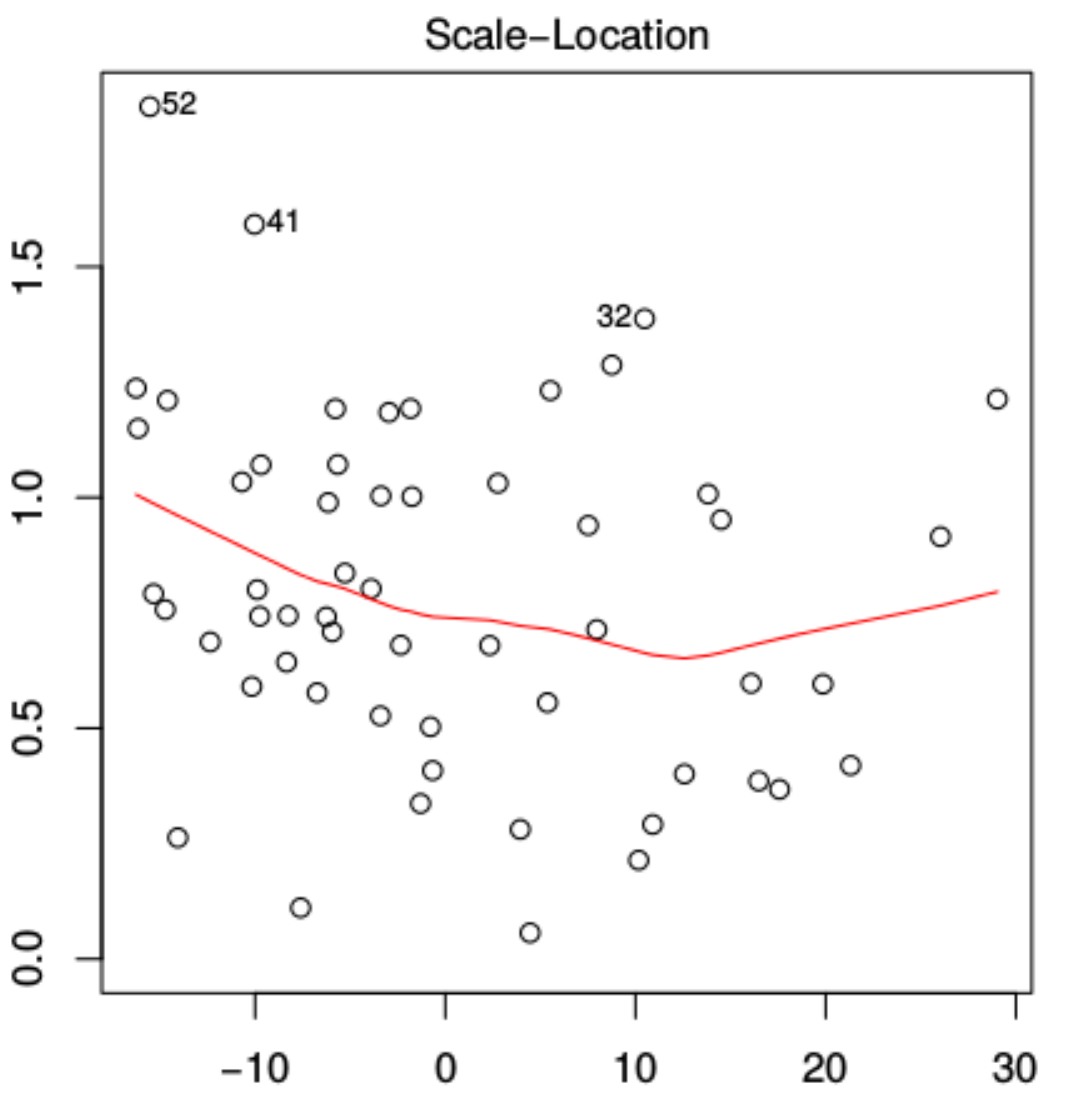}} \qquad
\subfigure{\label{Fig:PartialRegression} \includegraphics[width=0.25\linewidth]{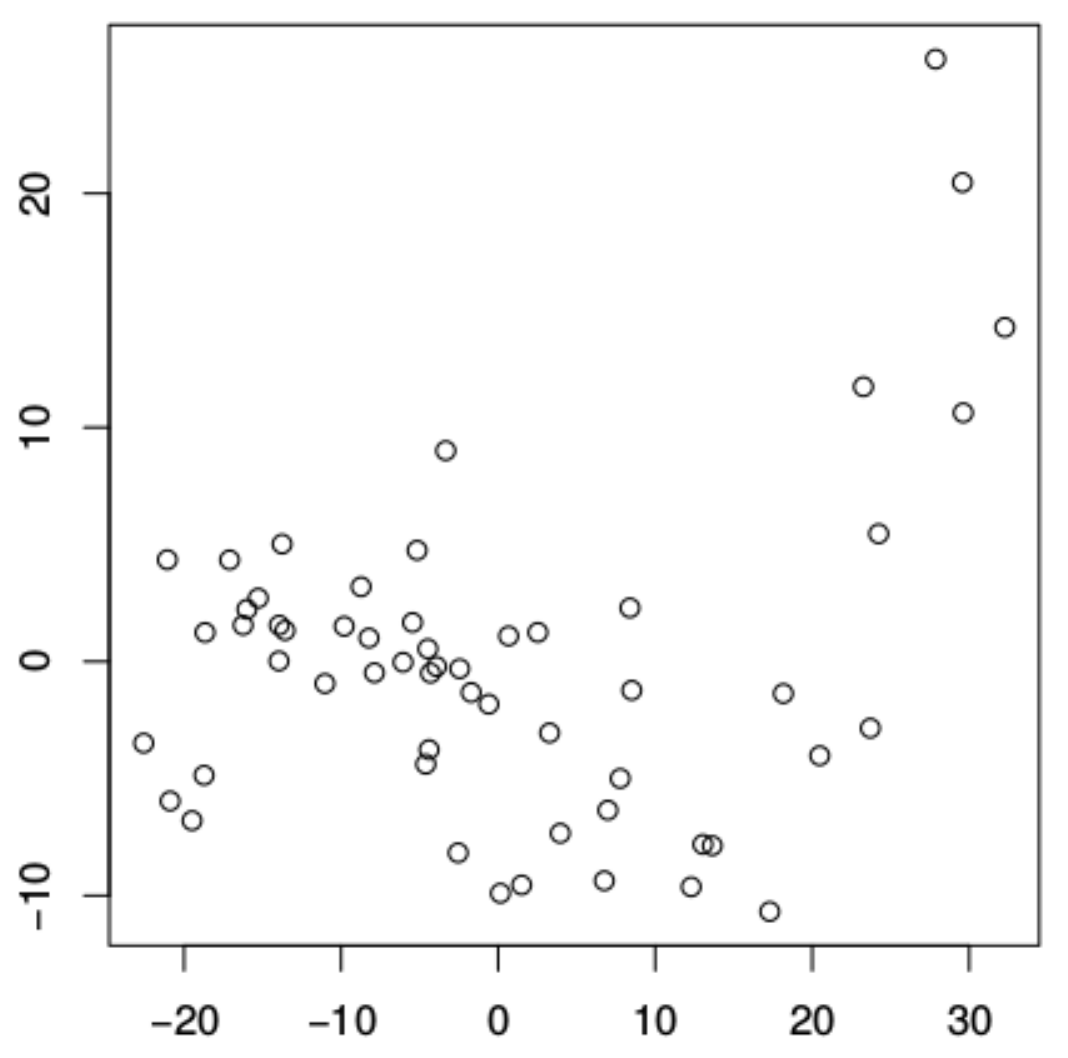}}
\caption{Plots of residuals against fitted values (left) and square roots of the absolute values of the standardised residuals against fitted values (middle) for the US cities temperature data.  The right panel gives a partial regression plot after removing the effect of latitude from both the response and longitude.}
\label{3figs}
\end{figure}
Based on this evidence we fitted a new model with a quadratic and cubic term in Longitude added to the first model. The $p$-value of the resulting $F$-test was $<2.2 \times 10^{-16}$, giving overwhelming evidence in favour of the inclusion of the quadratic and cubic terms. We then ran our extension of \texttt{MINTregression} to test the independence of $\epsilon$ and $(\mathrm{Longitude},\mathrm{Latitude})$, as described at the end of Section~\ref{Sec:Regression} with $B=1000$.  This resulted in a $p$-value of 0.0679, so we do not reject this model at the $5\%$ significance level. 

\section{Proofs}
\label{Sec:Proofs}

\subsection{Proofs from Section~\ref{Sec:KnownMarginals}}

\begin{proof}[Proof of Lemma~\ref{Lemma:Size}]
Under $H_0$, the finite sequence $(\hat{I}_n, \hat{I}_n^{(1)}, \ldots, \hat{I}_n^{(B)})$ is exchangeable.  If follows that under $H_0$, if we split ties at random, then every possible ordering of $\hat{I}_n, \hat{I}_n^{(1)}, \ldots, \hat{I}_n^{(B)}$ is equally likely.  In particular, the (descending order) rank of $\hat{I}_n$ among $\{\hat{I}_n, \hat{I}_n^{(1)}, \ldots, \hat{I}_n^{(B)}\}$, denoted $\mathrm{rank}(\hat{I}_n)$, is uniformly distributed on $\{1,2,\ldots,B+1\}$.  We deduce that for any $k \in \{1,\ldots,n-1\}$ and writing $\mathbb{P}_{H_0}$ for the probability under the product of any marginal distributions for $X$ and~$Y$, 
\begin{align*}
\mathbb{P}_{H_0}\bigl(\hat{I}_n > \hat{C}_q^{(n),B}\bigr) &= \mathbb{P}_{H_0}\biggl(\frac{1}{B+1}\sum_{b=0}^B \mathbbm{1}_{\{\hat{I}_n^{(b)} \geq \hat{I}_n\}} \leq q\biggr) \\
&\leq \mathbb{P}_{H_0}\bigl(\mathrm{rank}(\hat{I}_n) \leq q(B+1)\bigr) = \frac{\lfloor q(B+1) \rfloor}{B+1} \leq q,
\end{align*}
as required.
\end{proof}
Let $V(f) := \mathrm{Var} \log \frac{f(X,Y)}{f_X(X)f_Y(Y)}$.  The following result will be useful in the proof of Theorem~\ref{Thm:Power}.
\begin{lemma}
\label{Prop:Varlog}
Fix $d_X, d_Y \in \mathbb{N}$ and $\vartheta \in \Theta^2$. Then
\[
	\sup_{f \in \mathcal{F}_{d_X, d_Y, \vartheta}(0)} \frac{V(f)}{I(f)^{1/4}} < \infty.
\]
\end{lemma}

\begin{proof}[Proof of Lemma~\ref{Prop:Varlog}]
For $x \in \mathbb{R}$ we write $x_-:=\max(0,-x)$.  We have 
\begin{align*}
	\mathbb{E} \biggl[\biggl\{ \log \frac{f(X,Y)}{f_X(X)f_Y(Y)} \biggr\}_-\biggr] &= \int_{\{(x,y):f(x,y) \leq f_X(x)f_Y(y)\}} f(x,y) \log \frac{f_X(x)f_Y(y)}{f(x,y)} \,d\lambda_d(x,y) \\
&\leq \int_{\{(x,y):f(x,y) \leq f_X(x)f_Y(y)\}} f(x,y) \biggl\{ \frac{f_X(x)f_Y(y)}{f(x,y)} -1 \biggr\} \,d\lambda_d(x,y) \\
	& = \sup_{A \in \mathcal{B}(\mathbb{R}^d)} \biggl\{\int_A f_Xf_Y - \int_A f\biggr\} \leq \{ I(f) /2 \}^{1/2},
\end{align*}
where the final inequality is Pinsker's inequality.
We therefore have that
\begin{align*}
	V(f) &\leq \mathbb{E} \log^2 \frac{f(X,Y)}{f_X(X)f_Y(Y)} = \mathbb{E} \biggl\{ \Bigl| \log \frac{f(X,Y)}{f_X(X)f_Y(Y)} \Bigr|^{1/2} \Bigl| \log \frac{f(X,Y)}{f_X(X)f_Y(Y)} \Bigr|^{3/2} \biggr\} \nonumber \\
	& \leq \biggl\{ \mathbb{E} \Bigl| \log \frac{f(X,Y)}{f_X(X)f_Y(Y)} \Bigr| \biggr\}^{1/2}\biggl\{\mathbb{E}\Bigl| \log \frac{f(X,Y)}{f_X(X)f_Y(Y)} \Bigr|^3\biggr\}^{1/2} \nonumber \\
&= \biggl( I(f) + 2\mathbb{E} \biggl[\biggl\{ \log \frac{f(X,Y)}{f_X(X)f_Y(Y)} \biggr\}_-\biggr] \biggr)^{1/2} \biggl\{\mathbb{E}\Bigl| \log \frac{f(X,Y)}{f_X(X)f_Y(Y)} \Bigr|^3\biggr\}^{1/2} \nonumber \\
&\leq 2\bigl\{I(f) + 2^{1/2}I(f)^{1/2}\bigr\}^{1/2}\bigl\{\mathbb{E} \bigl|\log \bigl(f_X(X)f_Y(Y)\bigr)\bigr|^3 +\mathbb{E} |\log f(X,Y)|^3 \bigr\}^{1/2}.
\end{align*}
By \citet[][Lemma~11(i)]{BSY2017}, we conclude that 
\[
\sup_{f \in \mathcal{F}_{d_X,d_Y,\vartheta}:I(f) \leq 1} \frac{V(f)}{I(f)^{1/4}} < \infty.
\]
The result follows from this fact, together with the observation that 
\begin{equation}
\label{Eq:VBounded}
	\sup_{f \in \mathcal{F}_{d_X,d_Y,\vartheta}} V(f) \leq 2\sup_{f \in \mathcal{F}_{d_X,d_Y,\vartheta}}\bigl\{\mathbb{E}\log^2 f(X,Y) + \mathbb{E}\log^2 \bigl(f_X(X)f_Y(Y)\bigr)\bigr\} < \infty,
\end{equation}
where the final bound follows from another application of \citet[][Lemma~11(i)]{BSY2017}.
\end{proof}

\begin{proof}[Proof of Theorem~\ref{Thm:Power}]
Fix $f \in \mathcal{F}_{d_X,d_Y,\vartheta}(b_n)$.  We have by two applications of Markov's inequality that
\begin{align}
\label{Eq:Power1}
	\mathbb{P}_f(\hat{I}_n \leq \hat{C}_q^{(n),B_n}) &\leq \frac{1}{q(B_n+1)} \{1 +B_n \mathbb{P}_f(\hat{I}_n^{(1)} \geq \hat{I}_n)\} \nonumber \\
	&\leq \frac{1}{q} \mathbb{P}_f\bigl( |\hat{I}_n-\hat{I}_n^{(1)} - I(f)| \geq I(f)\bigr) +\frac{1}{q(B_n+1)} \nonumber \\
	&\leq \frac{1}{qI(f)^2} \mathbb{E}_f[ \{\hat{I}_n-\hat{I}_n^{(1)} - I(f)\}^2]+\frac{1}{q(B_n+1)}.
\end{align}
It is convenient to define the following entropy estimators based on pseudo-data, as well as oracle versions of our entropy estimators: let $Z_i^{(1)} := (X_i^T,Y_i^{(1)T})^T$, $Z^{(1)} := (Z_1^{(1)},\ldots,Z_n^{(1)})^T$ and $Y^{(1)} := (Y_1^{(1)},\ldots,Y_n^{(1)})^T$ and set
\begin{alignat*}{2}
&\hat{H}_n^{Z,(1)} := \hat{H}_{n,k}^{Z,w^Z}(Z^{(1)}), \quad &&\hat{H}_n^{Y,(1)} := \hat{H}_{n,k}^{Y,w^Y}(Y^{(1)})  \\
&H_n^* := -\frac{1}{n}\sum_{i=1}^n \log f(X_i,Y_i), \quad &&H_n^{*,Y} := -\frac{1}{n}\sum_{i=1}^n \log f_Y(Y_i), \\
&H_n^{*,(1)} := -\frac{1}{n}\sum_{i=1}^n \log f_X(X_i)f_Y(Y_i^{(1)}), \quad &&H_n^{*,Y,(1)} := -\frac{1}{n}\sum_{i=1}^n \log f_Y(Y_i^{(1)}).
\end{alignat*}
Writing $I_n^* := \frac{1}{n} \sum_{i=1}^n \log \frac{f(X_i,Y_i)}{f_X(X_i)f_Y(Y_i)}$ we have by \citet[][Theorem~1]{BSY2017} that
\begin{align*}
	&\sup_{\substack{k_Y \in \{k_0^*, \ldots, k_Y^*\} \\ k \in \{k_0^*, \ldots, k^*\}}}\sup_{f \in \mathcal{F}_{d_X,d_Y,\vartheta}} n \mathbb{E}_f\{(\hat{I}_n-\hat{I}_n^{(1)} - I_n^*)^2 \} \\
&= \! \! \! \sup_{\substack{k_Y \in \{k_0^*, \ldots, k_Y^*\} \\ k \in \{k_0^*, \ldots, k^*\}}}\sup_{f \in \mathcal{F}_{d_X,d_Y,\vartheta}} \! \! \! \! n\mathbb{E}_f\biggl[\biggl\{\hat{H}_n^Y - \hat{H}_n - \hat{H}_n^{Y,(1)} + \hat{H}_n^{(1)} - \bigl(H_n^{*,Y} - H_n^* - H_n^{*,Y,(1)} + H_n^{*,(1)}\bigr)\biggr\}^2\biggr] \\
&\leq 4n \sup_{\substack{k_Y \in \{k_0^*, \ldots, k_Y^*\} \\ k \in \{k_0^*, \ldots, k^*\}}}\sup_{f \in \mathcal{F}_{d_X,d_Y,\vartheta}} \mathbb{E}_f\Bigl\{(\hat{H}_n^Y - H_n^{*,Y})^2 + (\hat{H}_n - H_n^*)^2 \\
&\hspace{7cm}+ (\hat{H}_n^{Y,(1)} - H_n^{*,Y,(1)})^2 + (\hat{H}_n^{(1)} - H_n^{*,(1)})^2\Bigr\} \rightarrow 0.
\end{align*}
It follows by Cauchy--Schwarz and the fact that $\mathbb{E}\bigl[\bigl\{I_n^* - I(f)\bigr\}^2\bigr] = V(f)/n$ that
\begin{align}
\label{Eq:Epsilon3}
	\epsilon_n^2&:= \sup_{\substack{k_Y \in \{k_0^*, \ldots, k_Y^*\} \\ k \in \{k_0^*, \ldots, k^*\}}} \sup_{f \in \mathcal{F}_{d_X,d_Y,\vartheta}} \Bigl| n\mathbb{E}_f \bigl[\bigl\{\hat{I}_n- \hat{I}_n^{(1)}-I(f)\bigr\}^2 \bigr] - V(f) \Bigr| \nonumber \\
	& \leq \sup_{\substack{k_Y \in \{k_0^*, \ldots, k_Y^*\} \\ k \in \{k_0^*, \ldots, k^*\}}} \sup_{f \in \mathcal{F}_{d_X,d_Y,\vartheta}}  \Bigl\{ n \mathbb{E}_f \{(\hat{I}_n - \hat{I}_n^{(1)}- I_n^*)^2 \} + 2 \bigl[ n \mathbb{E}_f \{(\hat{I}_n - \hat{I}_n^{(1)} - I_n^*)^2 \} V(f) \bigr]^{1/2} \Bigr\} \nonumber \\
&\rightarrow 0,
\end{align}
where we use~\eqref{Eq:VBounded} to bound $V(f)$ above. Now consider $b_n:= \max( \epsilon_n^{1/2} n^{-1/2}, n^{-4/7} \log n)$. By~\eqref{Eq:Power1},~\eqref{Eq:Epsilon3} and Lemma~\ref{Prop:Varlog} we have that
\begin{align*}
	\inf_{B_n \geq B_n^*} \inf_{\substack{k_Y \in \{k_0^*, \ldots, k_Y^*\} \\ k \in \{k_0^*, \ldots, k^*\}}}&\inf_{f \in \mathcal{F}_{d_X,d_Y,\vartheta}(b_n)} \mathbb{P}_f(\hat{I}_n > C_q^{(n),B_n}) \\
	&\geq 1- \sup_{\substack{k_Y \in \{k_0^*, \ldots, k_Y^*\} \\ k \in \{k_0^*, \ldots, k^*\}}}\sup_{f \in \mathcal{F}_{d_X,d_Y,\vartheta}(b_n)} \frac{\mathbb{E}[ \{\hat{I}_n-\hat{I}_n^{(1)} - I(f)\}^2]}{qI(f)^2} -\frac{1}{q(B_n^*+1)} \\
	& \geq 1 - \sup_{f \in \mathcal{F}_{d_X,d_Y,\vartheta}(b_n)} \frac{V(f) + \epsilon_n^2}{nq I(f)^2} -\frac{1}{q(B_n^*+1)} \\
	& \geq 1 - \frac{1}{q \log^{7/4} n} \sup_{f \in \mathcal{F}_{d_X,d_Y,\vartheta}(0)} \frac{V(f)}{I(f)^{1/4}} - \frac{\epsilon_n}{q} -\frac{1}{q(B_n^*+1)} \rightarrow 1,
\end{align*}
as required.
\end{proof}

\subsection{Proofs from Section~\ref{Sec:UnknownMarginals}}

\begin{proof}[Proof of Lemma~\ref{Thm:size}]
We first claim that $(\hat{I}_n^{(0)}, \hat{I}_n^{(1)}, \ldots, \hat{I}_n^{(B)})$ is an exchangeable sequence under $H_0$. Indeed, let $\sigma$ be an arbitrary permutation of $\{0,1,\ldots,B\}$, and recall that $\hat{I}_n^{(b)}$ is computed from $(X_i,Y_{\tau_b(i)})_{i=1}^n$, where $\tau_b$ is uniformly distributed on $S_n$.  Note that, under $H_0$, we have $(X_i,Y_i)_{i=1}^n \overset{d}{=} (X_i,Y_{\tau(i)})_{i=1}^n$ for any $\tau \in S_n$.  Hence, for any $A \in \mathcal{B}(\mathbb{R}^{B+1})$,
\begin{align*}
	&\mathbb{P}_{H_0}\bigl((\hat{I}_n^{(\sigma(0))}, \ldots, \hat{I}_n^{(\sigma(B))}) \in A\bigr) \\
	&= \frac{1}{(n!)^B} \sum_{\pi_1, \ldots, \pi_B \in S_n} \mathbb{P}_{H_0}\bigl((\hat{I}_n^{(\sigma(0))}, \ldots, \hat{I}_n^{(\sigma(B))}) \in A | \tau_1=\pi_1, \ldots, \tau_B=\pi_B\bigr) \\
	&=\frac{1}{(n!)^B} \sum_{\pi_1, \ldots, \pi_B \in S_n} \mathbb{P}_{H_0}\bigl((\hat{I}_n^{(0)}, \ldots, \hat{I}_n^{(B)}) \in A | \tau_1=\pi_{\sigma(1)} \pi_{\sigma(0)}^{-1}, \ldots, \tau_B=\pi_{\sigma(B)} \pi_{\sigma(0)}^{-1}\bigr) \\
	&= \mathbb{P}_{H_0}\bigl((\hat{I}_n^{(0)}, \ldots, \hat{I}_n^{(B)}) \in A\bigr).
\end{align*}
By the same argument as in the proof of Lemma~\ref{Lemma:Size} we now have that
\[
	\mathbb{P}_{H_0}(\hat{I}_n > \tilde{C}_q^{(n),B}) \leq \frac{\lfloor q(B+1) \rfloor}{B+1} \leq q,
\]
as required.
\end{proof}
Recall the definition of $\hat{H}_n^{(1)}$ given just after~\eqref{Eq:Reject}.  We give the proof of Theorem~\ref{Thm:MINTunknownpower} below, followed by Lemma~\ref{Lemma:Unbiased}, which is the key ingredient on which it is based.
\begin{proof}[Proof of Theorem~\ref{Thm:MINTunknownpower}]
We have by Markov's inequality that
\begin{align}
\label{Eq:Power}
	\mathbb{P}_f(\hat{I}_n \leq \tilde{C}_q^{(n),B_n}) &\leq \frac{1}{q(B_n+1)} \{1 +B_n \mathbb{P}_f(\hat{I}_n \leq \tilde{I}_n^{(1)})\} \leq \frac{1}{q} \mathbb{P}_f(\hat{H}_n^Z \geq \tilde{H}_n^{(1)}) +\frac{1}{q(B_n+1)}.
\end{align}
But
\begin{align}
\label{Eq:Expansion}
	\mathbb{P}_f(\hat{H}_n^{Z}& \geq \hat{H}_n^{(1)}) = \mathbb{P}_f\bigl( \hat{H}_n^{Z} - H(X,Y) \geq \hat{H}_n^{(1)} - H(X)-H(Y)+I(X;Y) \bigr) \nonumber \\
	&\leq \mathbb{P}\Bigl(|\hat{H}_n^{Z}-H(X,Y)| \geq \frac{1}{2}I(X;Y) \Bigr) + \mathbb{P}\Bigl(|\hat{H}_n^{(1)} - H(X)-H(Y)| \geq \frac{1}{2}I(X;Y) \Bigr) \nonumber \\
&\leq \frac{2}{I(X;Y)}\bigl\{\mathbb{E}|\hat{H}_n^{Z}-H(X,Y)| + \mathbb{E}|\hat{H}_n^{(1)} - H(X)-H(Y)|\bigr\}.
\end{align}
The theorem therefore follows immediately from~\eqref{Eq:Power},~\eqref{Eq:Expansion} and Lemma~\ref{Lemma:Unbiased} below.
\end{proof}
\begin{lemma}
\label{Lemma:Unbiased}
Let $d_X, d_Y \in \mathbb{N}$, let $d = d_X+d_Y$ and fix $\phi=(\alpha,\mu,\nu,(c_n),(p_n)) \in \Phi$ with $c_n \rightarrow 0$ and $p_n = o(1/\log n)$ as $n \rightarrow \infty$.  Let $k_0^* = k_{0,n}^*$ and $k_1^* = k_{1,n}^*$ denote two deterministic sequences of positive integers satisfying $k_0^* \leq k_1^*$, $k_0^*/ \log^2 n \rightarrow \infty$ and $(k_1^* \log^2n) /n \rightarrow 0$.  Then 
\begin{align*}
\sup_{k \in \{k_0^*, \ldots, k_1^*\}} \sup_{f \in \mathcal{G}_{d_X,d_Y,\phi}} &\mathbb{E}_{f}| \hat{H}_n^{(1)} - H(X) - H(Y) | \rightarrow 0, \\
\sup_{k \in \{k_0^*, \ldots, k_1^*\}} \sup_{f \in \mathcal{H}_{d,\phi}} &\mathbb{E}_{f}| \hat{H}_n^Z - H(X,Y)| \rightarrow 0
\end{align*}
as $n \rightarrow \infty$.
\end{lemma}
\begin{proof}
We prove only the first claim in Lemma~\ref{Lemma:Unbiased}, since the second claim involves similar arguments but is more straightforward because the estimator is based on an independent and identically distributed sample and no permutations are involved.  Throughout this proof, we write $a \lesssim b$ to mean that there exists $C > 0$, depending only on $d_X, d_Y \in \mathbb{N}$ and $\phi \in \Phi$, such that $a \leq Cb$.  Fix $f \in \mathcal{G}_{d_X,d_Y,\phi}$, where $\phi = (\alpha,\mu,\nu,(c_n),(p_n))$.  Write $\rho_{(k),i,(1)}$ for the distance from $Z_i^{(1)}$ to its $k$th nearest neighbour in the sample $Z_1^{(1)}, \ldots, Z_n^{(1)}$ and $\xi_i^{(1)}:=e^{-\Psi(k)}V_d(n-1) \rho_{(k),i,(1)}^d$ so that
\[
	\hat{H}_n^{(1)} = \frac{1}{n} \sum_{i=1}^n \log \xi_i^{(1)}.
\]
Recall that $S_n$ denotes the set of all permutations of $\{1,\ldots,n\}$, and define $S_n^l \subseteq S_n$ to be the set of permutations that fix $\{1,\ldots,l\}$ and have no fixed points among $\{l+1,\ldots,n\}$; thus $|S_n^l| = (n-l)!\sum_{r=0}^{n-l} (-1)^r/r!$.  For $\tau \in S_n$, write $\mathbb{P}_\tau(\cdot) :=\mathbb{P}( \cdot | \tau_1=\tau)$, and $\mathbb{E}_\tau(\cdot) :=\mathbb{E}( \cdot | \tau_1=\tau)$.  Let $l_n := \lfloor \log \log n \rfloor$.  For $i=1,\ldots,n$, let
\[
A'_i:=\{\text{The $k$ nearest neighbours of $Z_i^{(1)}$ are among $Z_{l+1}^{(1)}, \ldots, Z_n^{(1)}$}\},
\]
and let $A_i:=A'_i \cap \{Z_i^{(1)} \in \mathcal{W}_n\}$.  Using the exchangeability of $\xi_1^{(1)},\ldots,\xi_n^{(1)}$, our basic decomposition is as follows:
\begin{align}
\label{Eq:MainDecomp}
&\mathbb{E}\bigl|\hat{H}_n^{(1)} - H(X) - H(Y)| \nonumber \\
&\leq \frac{1}{n!}\sum_{l=0}^{l_n}\binom{n}{l} \sum_{\tau \in S_n^l} \biggl\{\frac{1}{n}\sum_{i=1}^l \mathbb{E}_\tau |\log \xi_i^{(1)}| + \mathbb{E}_\tau\biggl|\frac{1}{n}\sum_{i=l+1}^n \log \xi_i^{(1)} - H(X) - H(Y)\biggr|\biggr\} \nonumber \\
&\hspace{7cm}+ \frac{1}{n!}\sum_{l=l_n+1}^n\binom{n}{l} \sum_{\tau \in S_n^l} \mathbb{E}_\tau\bigl|\hat{H}_n^{(1)} - H(X) - H(Y)\bigr| \nonumber \\
&\leq \frac{1}{n!}\sum_{l=0}^{l_n}\binom{n}{l} \sum_{\tau \in S_n^l} \biggl\{\frac{1}{n}\sum_{i=1}^l \mathbb{E}_\tau |\log \xi_i^{(1)}| + \mathbb{E}_\tau\biggl|\frac{1}{n}\sum_{i=l+1}^n (\log \xi_i^{(1)})\mathbbm{1}_{A_i} - H(X) - H(Y)\biggr| \nonumber \\
&+ \frac{1}{n}\sum_{i=l+1}^n \mathbb{E}_\tau\bigl(|\log \xi_i^{(1)}|\mathbbm{1}_{\{(A_i^c \setminus (A_i')^c) \cup (A_i')^c\}}\bigr)\biggr\} + \frac{1}{n!}\sum_{l=l_n+1}^n\binom{n}{l} \sum_{\tau \in S_n^l} \mathbb{E}_\tau\bigl(|\hat{H}_n^{(1)}| + |H(X) + H(Y)|\bigr).
\end{align}
The rest of the proof consists of handling each of these four terms.

\textbf{Step 1:} We first show that
\begin{equation}
\label{Eq:Step1}
\sup_{k \in \{1,\ldots,n-1\}} \sup_{f \in \mathcal{G}_{d_X,d_Y,\phi}} \frac{1}{n!}\sum_{l=l_n+1}^n\binom{n}{l} \sum_{\tau \in S_n^l} \mathbb{E}_\tau\bigl(|\hat{H}_n^{(1)}| + |H(X) + H(Y)|\bigr) \rightarrow 0
\end{equation}
as $n \rightarrow \infty$.  
Writing $\rho_{(k),i,X}$ for the distance from $X_i$ to its $k$th nearest neighbour in the sample $X_1, \ldots, X_n$ and defining $\rho_{(k),i,Y}$ similarly, we have
\[
	\max\{\rho_{(k),i,X}^2,\rho_{(k),\tau_1(i),Y}^2\}  \leq \rho_{(k),i,(1)}^2 \leq \rho_{(n-1),i,X}^2 + \rho_{(n-1),\tau_1(i),Y}^2.
\]
Using the fact that $\log(a+b) \leq \log 2 + |\log a| + |\log b|$ for $a,b>0$, we therefore have that
\begin{align}
	|& \log \xi_i^{(1)}| \leq \biggl| \log \biggl( \frac{V_d(n-1)}{e^{\Psi(k)}} \biggr) \biggr| + d |\log \rho_{(k),i,X}| + \frac{d}{2} \log \biggl( \frac{\rho_{(k),i,(1)}^2}{\rho_{(k),i,X}^2} \biggr)  \nonumber \\
	& \leq \biggl| \log \biggl( \frac{V_d(n-1)}{e^{\Psi(k)}} \biggr) \biggr| +2d|\log \rho_{(k),i,X}| + \frac{d}{2} \log 2 + d|\log \rho_{(n-1),i,X}|+ d|\log \rho_{(n-1),\tau_1(i),Y}| \label{Eq:crudebound1} \\
	& \leq \biggl| \log \biggl( \frac{V_d(n-1)}{e^{\Psi(k)}} \biggr) \biggr| + \frac{d}{2} \log 2 \nonumber \\
&\hspace{1cm} +4d \max_{j=1,\ldots,n} \max\{ -\log \rho_{(1),j,X} \, , \, -\log \rho_{(1),j,Y} \, , \, \log \rho_{(n-1),j,X} \, , \, \log \rho_{(n-1),j,Y} \}. \label{Eq:crudebound2}
\end{align}
Now, by the triangle inequality, a union bound and Markov's inequality,
\begin{align}
\label{Eq:max}
	\mathbb{E} \Bigl\{\max_{j=1, \ldots, n}&\log \rho_{(n-1),j,X} \Bigr\}- \log 2 \leq \mathbb{E}\log \Bigl( \max_{j=1,\ldots,n} \|X_j\| \Bigr) \leq \mathbb{E} \Bigl[\Bigl\{\log \Bigl( \max_{j=1,\ldots,n} \|X_j\|\Bigr)\Bigr\}_+\Bigr] \nonumber \\
	&= \int_0^\infty \mathbb{P} \Bigl( \max_{j=1,\ldots,n} \|X_j\| \geq e^M \Bigr) \,dM \nonumber \\
	& \leq \frac{1}{\alpha}\max\bigl\{0,\log n + \log \mathbb{E}( \|X_1\|^\alpha)\bigr\} + \frac{1}{\alpha} n \mathbb{E}( \|X_1\|^\alpha) \exp\bigl\{ - \log n - \log \mathbb{E}( \|X_1\|^\alpha)\bigr\} \nonumber \\
	& \leq \frac{1}{\alpha} \bigl\{1 + \log n + \max(0,\log \mu)\bigr\},
\end{align}
and the same bound holds for $\mathbb{E} \max_{j=1, \ldots, n} \log \rho_{(n-1),j,Y}$. Similarly, 
\begin{align}
\label{Eq:min}
	\mathbb{E}\Bigl[\Bigl\{\min_{j=1, \ldots, n} \log \rho_{(1),j,X} \Bigr\}_-\Bigr] &= \int_0^\infty \mathbb{P} \Bigl( \min_{j=1, \ldots, n} \rho_{(1),j,X} \leq e^{-M} \Bigr) \,dM \nonumber \\
	& \leq 2 d_X^{-1} \log n + n(n-1)V_{d_X} \|f_X\|_\infty \int_{2d_X^{-1}\log n}^\infty e^{-Md_X} \,dM \nonumber \\
	& \leq 2 d_X^{-1} \log n + d_X^{-1}V_{d_X}\nu,
\end{align}
and the same bound holds for $\mathbb{E}\bigl[\bigl\{\min_{j=1, \ldots, n} \log \rho_{(1),j,Y} \bigr\}_-\bigr]$ once we replace $d_X$ with $d_Y$.  By~\eqref{Eq:crudebound2}, \eqref{Eq:max}, \eqref{Eq:min} and  
Stirling's approximation, 
\begin{align}
\label{Eq:Stirling}
\sum_{l=l_n+1}^n &\binom{n}{l}\max_{k=1,\ldots,n-1} \sup_{f \in \mathcal{G}_{d_X,d_Y,\phi}} \mathbb{E}(|\hat{H}_n^{(1)}| \mathbbm{1}_{\{\tau_1 \in S_n^l \}}) \lesssim \log n \sum_{l=l_n+1}^n \binom{n}{l}\mathbb{P}(\tau_1 \in S_n^l)  \nonumber \\
&\leq \frac{\log n}{n!}\binom{n}{l_n+1} (n-l_n-1)! = \frac{\log n}{(l_n+1)!} \leq \frac{\log n}{\sqrt{2 \pi (l_n+1)}} \biggl( \frac{e}{l_n+1} \biggr)^{l_n+1} \rightarrow 0
\end{align}
as $n \rightarrow \infty$.  Moreover, for any $f \in \mathcal{G}_{d_X,d_Y,\phi}$ and $(X,Y) \sim f$, writing $\epsilon = \alpha/(\alpha+2d_X)$,
\begin{align}
\label{Eq:HXBound}
|H(X)| \leq \int_{\mathcal{X}} f_X|\log f_X| &\leq \int_{\mathcal{X}} f_X(x) \log\biggl(\frac{\|f_X\|_\infty}{f_X(x)}\biggr) \, dx + |\log \|f_X\|_\infty| \nonumber \\
&\leq \frac{\|f_X\|_\infty^\epsilon}{\epsilon} \int_{\mathcal{X}} f_X(x)^{1-\epsilon} \, dx + |\log \|f_X\|_\infty| \nonumber \\
	& \leq \frac{\nu^{\epsilon}}{\epsilon} (1+\mu)^{1-\epsilon} \biggl\{ \int_{\mathbb{R}^{d_X}} (1+\|x\|^\alpha)^{-\frac{1-\epsilon}{\epsilon}} \,dx \biggr\}^{\epsilon} \nonumber \\
&\hspace{1.5cm}+ \max\biggl\{\log \nu \, , \, \frac{1}{\alpha}\log\biggl(\frac{V_{d_X}^\alpha\mu^{d_X}(\alpha+d_X)^{\alpha+d_X}}{\alpha^\alpha d_X^{d_X}}\biggr)\biggr\},
\end{align}
where the lower bound on $\|f_X\|_\infty$ is due to the fact that $V_{d_X}^\alpha \mu_\alpha(f_X)^{d_X} \|f_X\|_\infty^\alpha \geq \alpha^\alpha d_X^{d_X}/(\alpha+d_X)^{\alpha+d_X}$ by \citet[][Lemma~10(i)]{BSY2017}.  Combining~\eqref{Eq:HXBound} with the corresponding bound on $H(Y)$, as well as~\eqref{Eq:Stirling}, we deduce that~\eqref{Eq:Step1} holds.

\medskip

\textbf{Step 2:} We observe that by~\eqref{Eq:crudebound2},~\eqref{Eq:max} and~\eqref{Eq:min},
\begin{equation}
\label{Eq:Step2}
\sup_{k \in \{1,\ldots,n-1\}} \sup_{f \in \mathcal{G}_{d_X,d_Y,\phi}} \frac{1}{n!}\sum_{l=0}^{l_n}\binom{n}{l} \sum_{\tau \in S_n^l} \frac{1}{n}\sum_{i=1}^l \mathbb{E}_\tau |\log \xi_i^{(1)}| \lesssim \frac{l_n\log n}{n} \rightarrow 0.
\end{equation}

\medskip

\textbf{Step 3:} We now show that
\begin{equation}
\label{Eq:Step3}
\sup_{k \in \{k_0^*,\ldots,k_1^*\}} \sup_{f \in \mathcal{G}_{d_X,d_Y,\phi}} \frac{1}{n!}\sum_{l=0}^{l_n}\binom{n}{l} \sum_{\tau \in S_n^l}\frac{1}{n}\sum_{i=l+1}^n \mathbb{E}_\tau\bigl(|\log \xi_i^{(1)}|\mathbbm{1}_{\{(A_i^c \setminus (A_i')^c) \cup (A_i')^c\}}\bigr) \rightarrow 0.
\end{equation}
We use an argument that involves covering $\mathbb{R}^d$ by cones; cf. \citet[][Section~20.7]{BiauDevroye2015}.  For $w_0 \in \mathbb{R}^d$, $z \in \mathbb{R}^d \setminus \{0\}$ and $\theta \in [0,\pi]$, define the cone of angle $\theta$ centred at $w_0$ in the direction $z$ by 
\[
	\mathcal{C}(w_0,z,\theta) := w_0 + \bigl\{ w \in \mathbb{R}^d \setminus \{0\}: \cos^{-1}(z^T w/(\|z\|\|w\|)) \leq \theta\bigr\} \cup \{0\}.
\]
By \citet[][Theorem~20.15]{BiauDevroye2015}, there exists a constant $C_{\pi/6} \in \mathbb{N}$ depending only on $d$ such that we may cover $\mathbb{R}^d$ by $C_{\pi/6}$ cones of angle $\pi/6$ centred at $Z_1^{(1)}$.  In each cone, mark the $k$ nearest points to $Z_1^{(1)}$ among $Z_2^{(1)}, \ldots, Z_n^{(1)}$.  Now fix a point $Z_i^{(1)}$ that is not marked and a cone that contains it, and let $Z_{i_1}^{(1)}, \ldots, Z_{i_k}^{(1)}$ be the $k$ marked points in this cone. By \citet[][Lemma~20.5]{BiauDevroye2015} we have, for each $j=1, \ldots, k$, that
\[
	\|Z_i^{(1)}- Z_{i_j}^{(1)} \| < \|Z_i^{(1)} - Z_1^{(1)}\|.
\]
Thus, $Z_1^{(1)}$ is not one of the $k$ nearest neighbours of the unmarked point $Z_i^{(1)}$, and only the marked points, of which there are at most $kC_{\pi/6}$, may have $Z_1^{(1)}$ as one of their $k$ nearest neighbours. This immediately generalises to show that at most $klC_{\pi/6}$ of the points $Z_{l+1}^{(1)}, \ldots, Z_n^{(1)}$ may have any of $Z_1^{(1)}, \ldots, Z_l^{(1)}$ among their $k$ nearest neighbours.  Then, by~\eqref{Eq:crudebound2}, \eqref{Eq:max} and~\eqref{Eq:min}, we have that, uniformly over $k \in \{1,\ldots,n-1\}$,
\begin{equation}
\label{Eq:Z1Fixed}
\sup_{f \in \mathcal{G}_{d_X,d_Y,\phi}} \max_{\tau \in \cup_{l=0}^{l_n} S_n^l} \frac{1}{n} \mathbb{E}_\tau\biggl(\sum_{i=l+1}^n |\log \xi_i^{(1)}| \mathbbm{1}_{(A'_i)^c}\biggr) \lesssim \frac{k l_n}{n} \log n.
\end{equation}
Now, for $x \in \mathcal{X}$ and $r \in [0,\infty)$, let $h_x(r) := \int_{B_x(r)} f_X$, and, for $s \in [0,1)$, let $h_x^{-1}(s) := \inf\{r \geq 0:h_x(r) \geq s\}$.  We also write $\mathrm{B}_{k,n-k}(s) := \frac{\Gamma(n)}{\Gamma(k)\Gamma(n-k)}s^{k-1}(1-s)^{n-k-1}$ for $s \in (0,1)$.  Then, by~(4) and~(13) in \citet{BSY2017}, there exists $C > 0$, depending only on $d_X$ and $\phi$, such that, uniformly for $k \in \{1,\ldots,n-1\}$,
\begin{align}
\label{Eq:rhoX}
&\max_{i=l+1,\ldots,n}\mathbb{E}_\tau\bigl\{|\log \rho_{(k),i,X}|\mathbbm{1}_{\{Z_i^{(1)} \in \mathcal{W}_n^c\}}\bigr\} \nonumber \\
&= \int_{\mathcal{W}_n^c} f_X(x)f_Y(y) \int_0^1 \biggl|\log\biggl(\frac{V_{d_X}(n-1)h_x^{-1}(s)}{e^{\Psi(k)}}\biggr)\biggr| \mathrm{B}_{k,n-k}(s) \, ds \, d\lambda_d(x,y) \nonumber \\
&\lesssim \int_{\mathcal{W}_n^c} f_X(x)f_Y(y) \int_0^1 \bigl\{\log n - \log s - \log(1-s) + \log(1+C\|x\|)\bigr\}\mathrm{B}_{k,n-k}(s) \, ds \, d\lambda_d(x,y).
\end{align}
Now, given $\epsilon \in (0,1)$, by H\"older's inequality, and with $K_\epsilon := \max\{1,(\epsilon\alpha-1)\log 2,C^{\epsilon \alpha}\}/(\epsilon \alpha)$,
\begin{align}
\label{Eq:logterm}
\int_{\mathcal{W}_n^c} &f_X(x)f_Y(y)\log(1+C\|x\|) \, d\lambda_d(x,y) \leq K_\epsilon\int_{\mathcal{W}_n^c} f_X(x)f_Y(y)(1+\|x\|^{\epsilon\alpha}) \, d\lambda_d(x,y) \nonumber \\
&\leq K_\epsilon \biggl\{\int_{\mathcal{W}_n^c} f_X(x)f_Y(y)(1+\|x\|^\alpha) \, d\lambda_d(x,y)\biggr\}^{\epsilon}\biggl\{\int_{\mathcal{W}_n^c} f_X(x)f_Y(y) \, d\lambda_d(x,y)\biggr\}^{1-\epsilon} \nonumber \\
&\leq K_\epsilon(1+\mu)^\epsilon p_n^{1-\epsilon}.
\end{align}
We deduce from~\eqref{Eq:rhoX} and~\eqref{Eq:logterm} that for each $\epsilon \in (0,1)$ and uniformly for $k \in \{1,\ldots,n-1\}$,
\[
\max_{l+1=1,\ldots,n} \sup_{f \in \mathcal{G}_{d_X,d_Y,\phi}} \mathbb{E}_\tau\bigl\{|\log \rho_{(k),i,X}|\mathbbm{1}_{\{Z_i^{(1)} \in \mathcal{W}_n^c\}}\bigr\} \lesssim p_n\log n + K_\epsilon p_n^{1-\epsilon}.
\]
From similar bounds on the corresponding terms with $\rho_{(k),i,X}$ replaced with $\rho_{(n-1),i,X}$ and  $\rho_{(n-1),\tau_1(i),Y}$, we conclude from~\eqref{Eq:crudebound1} that for every $\epsilon \in (0,1)$ and uniformly for $k \in \{1,\ldots,n-1\}$,
\begin{equation}
\label{Eq:Zncomp}
\max_{i=l+1,\ldots,n} \sup_{f \in \mathcal{G}_{d_X,d_Y,\phi}} \mathbb{E}_\tau\bigl\{|\log \xi_i^{(1)}|\mathbbm{1}_{\{Z_i^{(1)} \in \mathcal{W}_n^c\}}\bigr\} \lesssim p_n\log n + K_\epsilon p_n^{1-\epsilon}.
\end{equation}
The claim~\eqref{Eq:Step3} follows from~\eqref{Eq:Z1Fixed} and~\eqref{Eq:Zncomp}.


\medskip

\textbf{Step 4:} Finally, we show that
\begin{equation}
\label{Eq:Step4}
\sup_{k \in \{k_0^*,\ldots,k_1^*\}} \sup_{f \in \mathcal{G}_{d_X,d_Y,\phi}} \frac{1}{n!}\sum_{l=0}^{l_n}\binom{n}{l} \sum_{\tau \in S_n^l}\mathbb{E}_\tau\biggl|\frac{1}{n}\sum_{i=l+1}^n (\log \xi_i^{(1)})\mathbbm{1}_{A_i} - H(X) - H(Y)\biggr| \rightarrow 0.
\end{equation}
To this end, we consider the following further decomposition:
\begin{align}
\label{Eq:SecondDecomp}
\mathbb{E}_\tau\biggl|&\frac{1}{n}\sum_{i=l+1}^n (\log \xi_i^{(1)})\mathbbm{1}_{A_i} - H(X) - H(Y)\biggr| \nonumber \\
&\leq \frac{1}{n}\sum_{i=l+1}^n \mathbb{E}_\tau\bigl|\bigl(\log \xi_i^{(1)}f_Xf_Y(Z_i^{(1)})\bigr)\mathbbm{1}_{A_i}\bigr| + \frac{1}{n}\mathbb{E}_\tau\sum_{i=l+1}^n \bigl|\bigl(\log f_Xf_Y(Z_i^{(1)})\bigr)\mathbbm{1}_{\{(A_i^c \setminus (A_i')^c) \cup (A_i')^c\}}\bigr| \nonumber \\
&+\mathbb{E}_\tau \biggl|\frac{1}{n}\sum_{i=l+1}^n \log f_Xf_Y(Z_i^{(1)}) - \frac{n-l}{n}\bigl\{H(X) + H(Y)\bigr\}\biggr| + \frac{l}{n}\bigl|H(X) + H(Y)\bigr|.
\end{align}
We deal first with the second term in~\eqref{Eq:SecondDecomp}.  Now, by the arguments in Step~3,
\begin{align}
\label{Eq:fXfYDecomp}
\frac{1}{n}\mathbb{E}_\tau\sum_{i=l+1}^n &\bigl|\bigl(\log f_Xf_Y(Z_i^{(1)})\bigr)\mathbbm{1}_{(A_i')^c}\bigr| \leq \frac{C_{\pi/6}kl}{n}\mathbb{E}_\tau\biggl(\max_{i=l+1,\ldots,n} |\log f_Xf_Y(Z_i^{(1)})|\biggr) \nonumber \\
&\leq \frac{C_{\pi/6}k_1^*l}{n}\biggl\{\mathbb{E}\biggl(\max_{i=l+1,\ldots,n} |\log f_X(X_i)|\biggr) + \mathbb{E}\biggl(\max_{i=l+1,\ldots,n} |\log f_Y(Y_i)|\biggr)\biggr\}.
\end{align}
Now, for any $s > 0$, by H\"older's inequality,
\begin{align*}
\mathbb{P}&\bigl\{f_X(X_1) \leq s\bigr\} \leq s^{\frac{\alpha}{2(\alpha+d_X)}} \int_{x:f(x) \leq s} f_X(x)^{1 -\frac{\alpha}{2(\alpha+d_X)}} \, dx \\
&\leq  s^{\frac{\alpha}{2(\alpha+d_X)}}(1+\mu)^{1 - \frac{\alpha}{2(\alpha+d_X)}} \biggl\{\int_{\mathbb{R}^{d_X}} \frac{1}{(1+\|x\|^\alpha)^{(\alpha+2d_X)/\alpha}} \, dx\biggr\}^{\frac{\alpha}{2(\alpha+d_X)}} =: \tilde{K}_{d_X,\phi}s^{\frac{\alpha}{2(\alpha+d_X)}}.
\end{align*}
Writing $t_* := -\frac{2(\alpha+d_X)}{\alpha} \log (n-l)$, we deduce that for $n \geq 3$,
\begin{align}
\label{Eq:ExpLog}
\mathbb{E}\biggl(\max_{i=l+1,\ldots,n} |\log f_X(X_i)|\biggr) &\leq \log \nu - \mathbb{E}\biggl(\min_{i=l+1,\ldots,n} \log f_X(X_i)\biggr) \nonumber \\
&\leq \log \nu + (n-l)\tilde{K}_{d_X,\phi} \int_{-\infty}^{t_*} e^{\frac{\alpha t}{2(\alpha+d_X)}} \, dt + \frac{2(\alpha+d_X)}{\alpha} \log n \nonumber \\
&= \log \nu + \frac{2(\alpha+d_X)}{\alpha}\tilde{K}_{d_X,\phi} + \frac{2(\alpha+d_X)}{\alpha} \log n.
\end{align}
Combining~\eqref{Eq:ExpLog} with the corresponding bound on $\mathbb{E}\bigl(\max_{i=l+1,\ldots,n} |\log f_Y(Y_i)|\bigr)$, which is obtained in the same way, we conclude from~\eqref{Eq:fXfYDecomp} that
\begin{equation}
\label{Eq:Step4a1}
\sup_{k \in \{k_0^*,\ldots,k_1^*\}} \sup_{f \in \mathcal{G}_{d_X,d_Y,\phi}}\frac{1}{n}\mathbb{E}_\tau\sum_{i=l+1}^n \bigl|\bigl(\log f_Xf_Y(Z_i^{(1)})\bigr)\mathbbm{1}_{(A_i')^c}\bigr| \lesssim \frac{k_1^*l_n \log n}{n} \rightarrow 0.
\end{equation}
Moreover, given any $\epsilon \in \bigl(0,(\alpha+2d)/\alpha\bigr)$, set $\epsilon' := \alpha \epsilon/(\alpha+2d)$.  Then by two applications of H\"older's inequality,
\begin{align}
\label{Eq:Holderagain}
\frac{1}{n}&\mathbb{E}_\tau\sum_{i=l+1}^n \bigl|\bigl(\log f_Xf_Y(Z_i^{(1)})\bigr)\mathbbm{1}_{\{Z_i^{(1)} \in \mathcal{W}_n^c\}}\bigr| \nonumber \\
&\leq \int_{\mathcal{W}_n^c} f_Xf_Y(z)|\log f_Xf_Y(z)| \, dz \leq \nu p_n + \frac{\nu^{\epsilon'}}{\epsilon'} \int_{\mathcal{W}_n^c} f_Xf_Y(z)^{1-\epsilon'} \, dz \nonumber \\
&\leq \nu p_n + \frac{\nu^{\epsilon'}}{\epsilon'} p_n^{1-\epsilon} \biggl(\int_{\mathcal{W}_n^c} f_Xf_Y(z)^{2d/(\alpha+2d)} \, dz\biggr)^{\epsilon} \nonumber \\
&\leq \nu p_n + \frac{\nu^{\epsilon'}}{\epsilon'}p_n^{1-\epsilon}\biggl[\bigl\{1 + \max(1,2^{\alpha-1})\mu\bigr\}^{2d/(\alpha+2d)} \biggl\{\int_{\mathbb{R}^d} \frac{1}{(1+\|z\|^\alpha)^{2d/\alpha}} \, dz\biggr\}^{\alpha/(\alpha+2d)}\biggr]^\epsilon \lesssim p_n^{1-\epsilon}.
\end{align}
From~\eqref{Eq:Step4a1} and~\eqref{Eq:Holderagain}, we find that 
\begin{equation}
\label{Eq:Step4a}
\sup_{k \in \{k_0^*,\ldots,k_1^*\}} \sup_{f \in \mathcal{G}_{d_X,d_Y,\phi}} \max_{\tau \in \cup_{l=0}^{l_n} S_n^l} \frac{1}{n}\mathbb{E}_\tau\sum_{i=l+1}^n \bigl|\bigl(\log f_Xf_Y(Z_i^{(1)})\bigr)\mathbbm{1}_{\{(A_i^c \setminus (A_i')^c) \cup (A_i')^c\}}\bigr| \rightarrow 0,
\end{equation}
which takes care of the second term in~\eqref{Eq:SecondDecomp}.

For the third term in~\eqref{Eq:SecondDecomp}, since $Z_i^{(1)}$ and $\bigl\{Z_i^{(j)}:j \notin \{i,\tau(i),\tau^{-1}(i)\}\bigr\}$ are independent, we have
\begin{align*}
\mathbb{E}_\tau \biggl|\frac{1}{n}\sum_{i=l+1}^n &\log f_Xf_Y(Z_i^{(1)}) - \frac{n-l}{n}\bigl\{H(X) + H(Y)\bigr\}\biggr| \leq \mathrm{Var}_\tau^{1/2}\biggl(\frac{1}{n}\sum_{i=l+1}^n \log f_Xf_Y(Z_i^{(1)})\biggr) \nonumber \\
&\leq \frac{3^{1/2}}{n^{1/2}}\mathrm{Var}_\tau^{1/2}\log f_Xf_Y(Z_n^{(1)}) \leq \frac{3^{1/2}}{n^{1/2}}\biggl[\bigl\{\mathbb{E} \log^2 f_X(X_1)\bigr\}^{1/2} + \bigl\{\mathbb{E} \log^2 f_Y(Y_1)\bigr\}^{1/2}\biggr].
\end{align*}
By a very similar argument to that in~\eqref{Eq:HXBound}, we deduce that
\begin{equation}
\label{Eq:Step4b}
\sup_{k \in \{k_0^*,\ldots,k_1^*\}} \sup_{f \in \mathcal{G}_{d_X,d_Y,\phi}} \max_{\tau \in \cup_{l=0}^{l_n} S_n^l}\mathbb{E}_\tau \biggl|\frac{1}{n}\sum_{i=l+1}^n \log f_Xf_Y(Z_i^{(1)}) - \frac{n-l}{n}\bigl\{H(X) + H(Y)\bigr\}\biggr| \rightarrow 0,
\end{equation}
which handles the third term in~\eqref{Eq:SecondDecomp}.

For the fourth term in~\eqref{Eq:SecondDecomp}, we simply observe that, by~\eqref{Eq:HXBound},
\begin{equation}
\label{Eq:Step4c}
\sup_{k \in \{k_0^*,\ldots,k_1^*\}} \sup_{f \in \mathcal{G}_{d_X,d_Y,\phi}} \max_{\tau \in \cup_{l=0}^{l_n} S_n^l} \frac{l_n}{n}|H(X) + H(Y)| \rightarrow 0.
\end{equation}

Finally, we turn to the first term in~\eqref{Eq:SecondDecomp}.  
%
Recalling the definition of $r_{z,f_Xf_Y,\delta}$ in~\eqref{Eq:r}, we define the random variables
\[
	N_{i,\delta} := \sum_{\overset{l+1 \leq j \leq n}{j \neq i}} \mathbbm{1}_{\bigl\{ \|Z_j^{(1)}-Z_i^{(1)}\| \leq r_{Z_i^{(1)}, f_Xf_Y,\delta} \bigr\}}.
\]
We study the bias and the variance of $N_{i,\delta}$.  For $\delta \in (0,2]$ and $z \in \mathcal{W}_n$, we have that
\begin{align}
\label{Eq:ExpLower}
	\mathbb{E}_\tau(N_{i,\delta} |Z_i^{(1)}=z)& \geq (n-l_n-3) \int_{B_z(r_{z,f_Xf_Y,\delta})} f_Xf_Y(w) \,dw \nonumber \\
	&\geq (n-l_n-3)V_dr_{z,f_Xf_Y,\delta}^df_Xf_Y(z)(1-c_n) \nonumber \\
&= \delta(1-c_n)e^{\Psi(k)}\biggl(\frac{n-l_n-3}{n-1}\biggr).
\end{align}
Similarly, also, for $\delta \in (0,2]$ and $z \in \mathcal{W}_n$,
\begin{align}
\label{Eq:ExpUpper}
	\mathbb{E}_\tau(N_{i,\delta} |Z_i^{(1)}=z) &\leq 2+ (n-l_n-3) \int_{B_z(r_{z,f_Xf_Y,\delta})} f_Xf_Y(w) \,dw \nonumber \\
&\leq 2 + \delta(1+c_n)e^{\Psi(k)}\biggl(\frac{n-l_n-3}{n-1}\biggr).
\end{align}
Note that if $j_2 \notin \{ j_1, \tau(j_1), \tau^{-1}(j_1) \}$ then for $\delta \in (0,2]$ and $z \in \mathcal{W}_n$,
\[
	\mathrm{Cov}_\tau\Bigl( \mathbbm{1}_{\bigl\{ \|Z_{j_1}^{(1)}-z\| \leq r_{z,f_Xf_Y,\delta}\bigr\}}, \mathbbm{1}_{\bigl\{ \|Z_{j_2}^{(1)}-z\| \leq r_{z,f_Xf_Y,\delta}\bigr\}} |Z_i^{(1)} = z\Bigr)=0.
\]
Also, for $j \notin \{ i, \tau(i), \tau^{-1}(i) \}$ we have
\[
	\mathrm{Var}_\tau \Bigl(\mathbbm{1}_{\bigl\{ \|Z_j^{(1)}-z\| \leq r_{z,f_Xf_Y,\delta}\bigr\}} | Z_i^{(1)} =z \Bigr) \leq \int_{B_z(r_{z,f_Xf_Y,\delta})} f_Xf_Y(w) \,dw \leq \frac{\delta(1+c_n)e^{\Psi(k)}}{n-1}.
\]
When $j \in \{ i, \tau(i), \tau^{-1}(i)\}$ we simply bound the variance above by $1/4$ so that, by Cauchy--Schwarz, when $n-1 \geq 2(1+c_ne^{\Psi(k)})$,
\begin{align}
\label{Eq:VarUpper}
	\mathrm{Var}_\tau (N_{i,\delta} |Z_i^{(1)}=z) \leq 3(n-l_n-3) \frac{\delta(1+c_n)e^{\Psi(k)}}{n-1} + 1.
\end{align}
Letting $\epsilon_{n,k} := \max(k^{-1/2}\log n,c_n^{1/2})$, there exists $n_0 \in \mathbb{N}$, depending only on $d_X, d_Y$, $\phi$ and $k_0^*$, such that for $n \geq n_0$ and all $k \in \{k_0^*,\ldots,k_1^*\}$, we have $\epsilon_{n,k} \leq 1/2$ and
\[
\min\biggl\{(1+\epsilon_{n,k})(1-c_n)e^{\Psi(k)}\Bigl(\frac{n\!-\!l_n\!-\!3}{n-1}\Bigr) - k \, , \, k - 2 - (1-\epsilon_{n,k})(1+c_n)e^{\Psi(k)}\Bigl(\frac{n\!-\!l_n\!-\!3}{n-1}\Bigr)\biggr\} \geq \frac{k\epsilon_{n,k}}{2}.
\]
We deduce from~\eqref{Eq:VarUpper} and~\eqref{Eq:ExpLower} that for $n \geq n_0$,
\begin{align}
\label{Eq:CIP}
	\mathbb{P}_\tau\bigl( \{\xi_i^{(1)} &f_Xf_Y(Z_i^{(1)})-1\} \mathbbm{1}_{A_i} \geq \epsilon_{n,k}\bigr) = \mathbb{P}_\tau\bigl(A_i \cap \{\rho_{(k),i,(1)} \geq r_{Z_i^{(1)}, f_Xf_Y,1+\epsilon_{n,k}}\}\bigr) \nonumber \\
&\leq \mathbb{P}_\tau (N_{i,1+\epsilon_{n,k}} \leq k, Z_i^{(1)} \in \mathcal{W}_n) = \int_{\mathcal{W}_n}f_Xf_Y(z) \mathbb{P}_\tau(N_{i,1+\epsilon_{n,k}} \leq k | Z_i^{(1)}=z) \,dz \nonumber \\
	&\leq \int_{\mathcal{W}_n} f_Xf_Y(z) \frac{\Var_\tau(N_{i,1+\epsilon_{n,k}} |Z_i^{(1)}=z)}{\{\mathbb{E}_\tau(N_{i,1+\epsilon_{n,k}} |Z_i^{(1)}=z) -k \}^2} \,dz \nonumber \\
&\leq \frac{3(n-l_n-3) \frac{(1+\epsilon_{n,k})(1+c_n)e^{\Psi(k)}}{n-1} + 1}{\bigl\{(1+\epsilon_{n,k})(1-c_n)e^{\Psi(k)}\bigl(\frac{n-l_n-3}{n-1}\bigr) - k \bigr\}^2}.
\end{align}
Using very similar arguments, but using~\eqref{Eq:ExpUpper} in place of~\eqref{Eq:ExpLower}, we also have that for $n \geq n_0$,
\begin{equation}
\label{Eq:CIP2}
\mathbb{P}_\tau\bigl( \{\xi_i^{(1)} f_Xf_Y(Z_i^{(1)})-1\} \mathbbm{1}_{A_i} \leq -\epsilon_{n,k}\bigr) \leq \frac{3(n-l_n-3) \frac{(1-\epsilon_{n,k})(1+c_n)e^{\Psi(k)}}{n-1} + 1}{\bigl\{k - 2 - (1-\epsilon_{n,k})(1+c_n)e^{\Psi(k)}\bigl(\frac{n-l_n-3}{n-1}\bigr)\bigr\}^2}.
\end{equation}
Now, by Markov's inequality, for $k \geq 3$, $i \in \{l+1,\ldots,n\}$, $\tau \in S_n^l$ and $\delta \in (0,2]$,
\begin{align}
\label{Eq:Smalldelta}
	\mathbb{P}_\tau\bigl(\bigl\{\xi_i^{(1)}f_Xf_Y(Z_i^{(1)}) \leq \delta\bigr\} \cap A_i\bigr) &= \mathbb{P}_\tau (N_{i,\delta} \geq k, Z_i^{(1)} \in \mathcal{W}_n) \nonumber \\
&\leq \frac{n-l-3}{k-2} \int_{\mathcal{W}_n} f_Xf_Y(z) \int_{B_z(r_{z,f_Xf_Y,\delta})} f_Xf_Y(w) \,dw \, dz \nonumber \\ 
&\leq \frac{n-l-3}{k-2}(1+c_n) \frac{\delta e^{\Psi(k)}}{n-1}.
\end{align}
Moreover, for any $i \in \{l+1,\ldots,n\}$, $\tau \in S_n^l$ and $t > 0$, by two applications of Markov's inequality,
\begin{align}
\label{Eq:ttail}
	\mathbb{P}_\tau\bigl(\bigl\{\xi_i^{(1)} f_Xf_Y(Z_i^{(1)}) \geq t\bigr\} \cap A_i\bigr) &\leq \mathbb{P}_\tau(N_{i,t} \leq k, Z_i^{(1)} \in \mathcal{W}_n) \nonumber \\
	&\leq \frac{n-l-3}{n-l-k-3} \int_{\mathcal{W}_n} f_Xf_Y(z)\int_{B_z(r_{z,f_Xf_Y,t})^c} f_Xf_Y(w) \,dw \,dz \nonumber \\
	&\leq \frac{n-l_n-3}{n-l_n-k-3} \max\{1,2^{\alpha-1}\}  \int_{\mathcal{W}_n }  f_Xf_Y(z) \frac{\mu +\|z\|^\alpha}{r_{z,f_Xf_Y,t}^\alpha} \,dz \nonumber \\
&\leq \frac{n-l_n-3}{n-l_n-k-3} \max\{1,2^{\alpha-1}\} 2\mu \nu^{2\alpha/d}V_d^{\alpha/d}\Bigl(\frac{n-1}{e^{\Psi(k)}}\Bigr)^{\alpha/d}t^{-\alpha/d}.
\end{align}
Writing $s_n := \log^2\bigl((n/k)\log^{2d/\alpha} n\bigr)$, we deduce from~\eqref{Eq:CIP},~\eqref{Eq:CIP2},~\eqref{Eq:Smalldelta} and~\eqref{Eq:ttail} that for $n \geq n_0$,
\begin{align*}
	&\mathbb{E}_\tau \bigl\{\log^2 \bigl( \xi_i^{(1)} f_Xf_Y(Z_i^{(1)}) \bigr) \mathbbm{1}_{A_i}\bigr\} = \int_0^\infty \mathbb{P}_\tau\bigr(\bigl\{\xi_i^{(1)} f_Xf_Y(Z_i^{(1)}) \leq e^{-s^{1/2}}\bigr\} \cap A_i\bigr) \, ds \\
&+ \log^2(1+\epsilon_{n,k}) + \biggl(\int_{\log^2(1+\epsilon_{n,k})}^{s_n} + \int_{s_n}^\infty\biggr) \mathbb{P}_\tau\bigl(\bigl\{\xi_i^{(1)} f_Xf_Y(Z_i^{(1)}) > e^{s^{1/2}}\bigr\} \cap A_i\bigr) \, ds \\
&\leq \frac{n-l-3}{k-2}(1+c_n) \frac{e^{\Psi(k)}}{n-1}\int_0^\infty e^{-s^{1/2}} \, ds + \log^2(1+\epsilon_{n,k}) \\
&\hspace{2cm}+ s_n\frac{12(n-l_n-3) \frac{(1+\epsilon_{n,k})(1+c_n)e^{\Psi(k)}}{n-1} + 1}{k^2\epsilon_{n,k}^2} \\
&\hspace{2cm}+ \frac{n-l_n-3}{n-l_n-k-3} \max\{1,2^{\alpha-1}\} 2\mu \nu^{2\alpha/d}V_d^{\alpha/d}\Bigl(\frac{n-1}{e^{\Psi(k)}}\Bigr)^{\alpha/d}\int_{s_n}^\infty e^{-s^{1/2}\alpha/d} \, ds.
\end{align*}
We conclude that
\begin{equation}
\label{Eq:FiveSups}
\sup_{n \in \mathbb{N}} \sup_{k \in \{k_0^*,\ldots,k_1^*\}} \sup_{f \in \mathcal{G}_{d_X,d_Y,\phi}} \max_{\tau \in \cup_{l=1}^{l_n} S_n^l} \max_{i=l+1,\ldots,n} \mathbb{E}_\tau \bigl\{\log^2 \bigl( \xi_i^{(1)} f_Xf_Y(Z_i^{(1)}) \bigr) \mathbbm{1}_{A_i}\bigr\} < \infty.
\end{equation}
Finally, then, from~\eqref{Eq:CIP},~\eqref{Eq:CIP2} and~\eqref{Eq:FiveSups},
\begin{align*}
&\frac{1}{n} \sum_{i=l+1}^n \mathbb{E}_\tau \bigl|\log \bigl(\xi_i^{(1)}f_Xf_Y(Z_i^{(1)})\bigr) \mathbbm{1}_{A_i}\bigr| \leq \max_{i=l+1,\ldots,n}\mathbb{E}_\tau\bigl| \log \bigl(\xi_i^{(1)}f_Xf_Y(Z_i^{(1)})\bigr) \mathbbm{1}_{A_i}\bigr| \\ 
&\leq 2\epsilon_{n,k} + \max_{i=l+1,\ldots,n} \Bigl[\mathbb{P}_\tau\bigl\{\bigl| \log \bigl(\xi_i^{(1)}f_Xf_Y(Z_i^{(1)})\bigr)\bigr|\mathbbm{1}_{A_i} \geq 2\epsilon_{n,k}\bigr\} \mathbb{E}_\tau\bigl\{\log^2 \bigl(\xi_i^{(1)}f_Xf_Y(Z_i^{(1)})\bigr) \mathbbm{1}_{A_i}\bigr\}\Bigr]^{1/2} \\
&\leq 2\epsilon_{n,k} + \max_{i=l+1,\ldots,n} \Bigl[\mathbb{P}_\tau\bigl\{\bigl|\xi_i^{(1)}f_Xf_Y(Z_i^{(1)}) - 1 \bigr|\mathbbm{1}_{A_i} \geq \epsilon_{n,k}\bigr\} \mathbb{E}_\tau\bigl\{\log^2 \bigl(\xi_i^{(1)}f_Xf_Y(Z_i^{(1)})\bigr) \mathbbm{1}_{A_i}\bigr\}\Bigr]^{1/2},
\end{align*}
so
\begin{equation}
\label{Eq:Step4d}
\sup_{k \in \{k_0^*,\ldots,k_1^*\}} \sup_{f \in \mathcal{G}_{d_X,d_Y,\phi}} \max_{\tau \in \cup_{l=1}^{l_n} S_n^l}\frac{1}{n} \sum_{i=l+1}^n \mathbb{E}_\tau \bigl|\log \bigl(\xi_i^{(1)}f_Xf_Y(Z_i^{(1)})\bigr) \mathbbm{1}_{A_i}\bigr| \rightarrow 0.
\end{equation}
as $n \rightarrow \infty$.  The proof of claim~\eqref{Eq:Step4} follows from~\eqref{Eq:SecondDecomp}, together with~\eqref{Eq:Step4a},~\eqref{Eq:Step4b},~\eqref{Eq:Step4c} and~\eqref{Eq:Step4d}.  

\medskip

The final result therefore follows from~\eqref{Eq:MainDecomp},~\eqref{Eq:Step1},~\eqref{Eq:Step2},~\eqref{Eq:Step3} and~\eqref{Eq:Step4}.
\end{proof}

\subsection{Proof from Section~\ref{Sec:Regression}}

\begin{proof}[Proof of Theorem~\ref{Thm:MINTregressionpower}]
We partition $X = (X_{(1)}^T \ X_{(2)}^T)^T \in \mathbb{R}^{(m+m) \times p}$ and, writing $\eta^{(0)} := \epsilon/\sigma$, let $\eta^{(b)} = \bigl((\eta_{(1)}^{(b)})^T,(\eta_{(2)}^{(b)})^T\bigr)^T \in \mathbb{R}^{m+m}$, for $b=0,\ldots,B$.  Further, let $\hat{\gamma} := (X_{(2)}^TX_{(2)})^{-1}X_{(2)}^T\eta_{(2)}^{(0)}$, $\hat{\gamma}^{(1)} := (X_{(2)}^TX_{(2)})^{-1}X_{(2)}^T\eta_{(2)}^{(1)}$ and $\hat{s} := \hat{\sigma}_{(2)}/\sigma$.  Now define the events 
\[
	A_{r_0,s_0}:=\Bigl\{ \max(\|\hat{\gamma}\|,\|\hat{\gamma}^{(1)}\|) \leq r_0 \, , \, \hat{s} \in [s_0,1/s_0] \, , \, \hat{s}^{(1)} \in [s_0,1/s_0] \Bigr\}
\]
and $A_0 := \{\lambda_\text{min}(n^{-1} X^TX) > \frac{1}{2} \lambda_\text{min}(\Sigma)\}$.  Now
\begin{equation}
\label{Eq:BreveDecomp}
	\mathbb{P}(\breve{I}_n^{(0)} \leq \breve{I}_n^{(1)}) \leq \mathbb{P}\bigl(\{\breve{I}_n^{(0)} \leq \breve{I}_n^{(1)}\} \cap A_{r_0,s_0} \cap A_0\bigr) + \mathbb{P}(A_0^c) + \mathbb{P}(A_0 \cap A_{r_0,s_0}^c).
\end{equation}
For the first term in~\eqref{Eq:BreveDecomp}, define functions $h, h^{(1)}:\mathbb{R}^p \rightarrow \mathbb{R}$ by 
\[
	h(b):=H(\eta_1 - X_1^Tb) \quad \text{and} \quad h^{(1)}(b):=H(\eta_1^{(1)}-X_1^Tb).
\]
Writing $R_1$ and $R_2$ for remainder terms to be bounded below
, we have
\begin{align}
\label{Eq:R1R2}
	\breve{I}_n^{(0)} - \breve{I}_n^{(1)} &= \hat{H}_m(\hat{\eta}_{1,(1)}, \ldots, \hat{\eta}_{m,(1)}) - \hat{H}_m \bigl( (X_1, \hat{\eta}_{1,(1)}), \ldots, (X_m, \hat{\eta}_{m,(1)}) \bigr) \nonumber \\
&\hspace{2cm}- \hat{H}_m(\hat{\eta}_{1,(1)}^{(1)}, \ldots, \hat{\eta}_{m,(1)}^{(1)}) + \hat{H}_m \bigl( (X_1, \hat{\eta}_{1,(1)}^{(1)}), \ldots, (X_m, \hat{\eta}_{m,(1)}^{(1)}) \bigr) \nonumber \\
&= I(X_1;\epsilon_1) + h( \hat{\gamma}) - h(0) - h^{(1)}(\hat{\gamma}^{(1)}) + h^{(1)}(0) + R_1 \nonumber \\
&= I(X_1;\epsilon_1) + R_1 + R_2.
\end{align}
To bound $R_1$ on the event $A_{r_0,s_0}$, we first observe that for fixed $\gamma,\gamma^{(1)} \in \mathbb{R}^p$ and $s,s^{(1)} > 0$,
\begin{align}
\label{Eq:4Hs}
H\biggl(\frac{\eta_1 - X_1^T\gamma}{s}\biggr) &- H\biggl(X_1,\frac{\eta_1 - X_1^T\gamma}{s}\biggr) - H\biggl(\frac{\eta_1^{(1)} - X_1^T\gamma^{(1)}}{s^{(1)}}\biggr) + H\biggl(X_1,\frac{\eta_1^{(1)} - X_1^T\gamma^{(1)}}{s^{(1)}}\biggr) \nonumber \\
&= H\biggl(\frac{\eta_1 - X_1^T\gamma}{s}\biggr) - H\Bigl(\frac{\eta_1}{s}\Bigm| X_1\Bigr) - H\biggl(\frac{\eta_1^{(1)} - X_1^T\gamma^{(1)}}{s^{(1)}}\biggr) + H\Bigl(\frac{\eta_1^{(1)}}{s^{(1)}}\Bigm| X_1\Bigr) \nonumber \\
&= H(\eta_1) - H(\eta_1|X_1) + h(\gamma) - h(0) - h^{(1)}(\gamma^{(1)}) + h^{(1)}(0) \nonumber \\
&= I(X_1;\epsilon_1) + h(\gamma) - h(0) - h^{(1)}(\gamma^{(1)}) + h^{(1)}(0).
\end{align}
Now, for a density $g$ on $\mathbb{R}^d$, define the variance functional
\[
v(g) := \int_{\mathbb{R}^d} g(x)\{\log g(x) + H(g)\}^2 \, dx.
\]
If $(\hat{\gamma},\hat{\gamma}^{(1)},\hat{s},\hat{s}^{(1)})^T$ are such that $A_{r_0,s_0}$ holds, then conditional on $(\hat{\gamma},\hat{\gamma}^{(1)},\hat{s},\hat{s}^{(1)})$, we have $f_{\hat{\eta}}^{\hat{\gamma},\hat{s}} \in \mathcal{F}_{1,\theta_1}$, $f_{X,\hat{\eta}}^{\hat{\gamma},\hat{s}} \in \mathcal{F}_{p+1,\theta_2}$, $f_{\hat{\eta}^{(1)}}^{\hat{\gamma}^{(1)}, \hat{s}^{(1)}} \in \mathcal{F}_{1,\theta_1}$ and $f_{X,\hat{\eta}^{(1)}}^{\hat{\gamma}^{(1)},\hat{s}^{(1)}} \in \mathcal{F}_{p+1,\theta_2}$.  It follows by~\eqref{Eq:4Hs}, \citet[][Theorem~1 and Lemma~11(i)]{BSY2017} that
\begin{align}
\label{Eq:R1Bound}
\limsup_{n \rightarrow \infty} & \ n^{1/2} \sup_{\substack{k_\eta \in \{k_0^*,\ldots,k_\eta^*\}\\k \in \{k_0^*,\ldots,k^*\}}} \sup_{f \in \mathcal{F}_{p+1,\omega}^*} \mathbb{E}_f(|R_1|\mathbbm{1}_{A_{r_0,s_0}}) \nonumber \\
&\leq \limsup_{n \rightarrow \infty} \ n^{1/2} \sup_{\substack{k_\eta \in \{k_0^*,\ldots,k_\eta^*\}\\k \in \{k_0^*,\ldots,k^*\}}} \sup_{f \in \mathcal{F}_{p+1,\omega}^*} \bigl\{\mathbb{E}(R_1^2\mathbbm{1}_{A_{r_0,s_0}})\bigr\}^{1/2} \nonumber \\
&\leq 4\sup_{g \in \mathcal{F}_{1,\theta_1}} v(g) + 4\sup_{g \in \mathcal{F}_{p+1,\theta_2}} v(g) < \infty.
\end{align}
Now we turn to $R_2$, and study the continuity of the functions $h$ and $h^{(1)}$, following the approach taken in Proposition~1 of \citet{Polyanskiy16}.  Write $a_1 \in \mathcal{A}$ for the fifth component of $\theta_1$, and for $x \in \mathbb{R}$ with $f_\eta(x) > 0$, let 
\[
r_{a_1}(x) := \frac{1}{8a_1\bigl(f_\eta(x)\bigr)}.
\]
Then, for $|y-x| \leq r_{a_1}(x)$ we have by \citet[][Lemma~12]{BSY2017} that
\[
	|f_{\eta}(y) - f_{\eta}(x) | \leq \frac{15}{7} f_{\eta}(x)a_1(f_{\eta}(x)) |y-x| \leq \frac{15}{56} f_{\eta}(x).
\]
Hence
\[
	|\log f_{\eta}(y) - \log f_{\eta}(x) | \leq \frac{120}{41} a_1(f_{\eta}(x)) |y-x|.
\]
When $|y-x| > r_{a_1}(x)$ we may simply write
\[
	|\log f_{\eta}(y) - \log f_{\eta}(x) | \leq 8\{ | \log  f_{\eta}(y) | + | \log  f_{\eta}(x) | \}  a_1( f_{\eta}(x)) |y-x|.
\]
Combining these two equations we now have that, for any $x,y$ such that $f_{\eta}(x) >0$, 
\[
	|\log f_{\eta}(y) - \log f_{\eta}(x) | \leq 8\{1+ | \log f_{\eta}(y) | + | \log  f_{\eta}(x) | \} a_1( f_{\eta}(x)) |y-x|.
\]
By an application of the generalised H\"older inequality and Cauchy--Schwarz, we conclude that
\begin{align}
\label{Eq:GenHolder}
	\mathbb{E} &\biggl| \log \frac{ f_{\eta}(\eta_1 - X_1^T\gamma)}{ f_{\eta}(\eta_1)} \biggr| \nonumber \\
&\leq 8 \bigl\{\mathbb{E}a_1^2(f_\eta(\eta_1))\bigr\}^{1/2}\bigl[\mathbb{E}\bigl\{\bigl(1+ | \log f_{\eta}(\eta_1-X_1^T\gamma) | +  | \log f_{\eta}(\eta_1) |\bigr)^2\bigr\}\bigr]^{1/2}\bigl\{\mathbb{E}(|X_1^T\gamma|^4\bigr)\bigr\}^{1/4} \nonumber \\
	& \leq 16\|\gamma\|\bigl\{\mathbb{E}a_1^2(f_\eta(\eta_1))\bigr\}^{1/2}\bigl[\mathbb{E}\bigl\{1+ \log^2 f_{\eta}(\eta_1-X_1^T\gamma) +  \log^2 f_{\eta}(\eta_1)\bigr\}\bigr]^{1/2} \bigl\{\mathbb{E}(\|X_1\|^4)\bigr\}^{1/4}.
\end{align}
We also obtain a similar bound on the quantity $\mathbb{E} \Bigl| \log \frac{ f_{\hat{\eta}}^{\gamma,1}(\eta_1 - X_1^T\gamma)}{ f_{\hat{\eta}}^{\gamma,1}(\eta_1)} \Bigr|$ when $\|\gamma\| \leq r_0$.  Moreover, for any random vectors $U,V$ with densities $f_U,f_V$ on $\mathbb{R}^d$ satisfying $H(U),H(V) \in \mathbb{R}$, and writing $\mathcal{U} := \{f_U > 0\}$, $\mathcal{V} := \{f_V > 0\}$, we have by the non-negativity of Kullback--Leibler divergence that
\begin{align*}
H(U) - H(V) = \int_{\mathcal{V}} f_V \log f_V - \int_{\mathcal{U}} f_U \log f_V - \int_{\mathcal{U}} f_U \log \frac{f_U}{f_V} \leq \mathbb{E}\log \frac{f_V(V)}{f_V(U)}.
\end{align*}
Combining this with the corresponding bound for $H(V) - H(U)$, we obtain that
\[
|H(U) - H(V)| \leq \max\biggl\{\mathbb{E}\Bigl|\log \frac{f_U(U)}{f_U(V)}\Bigr| \, , \, \mathbb{E}\Bigl|\log \frac{f_V(V)}{f_V(U)}\Bigr|\biggr\}.
\]
Since $f_{\eta},f_{\hat{\eta}}^{\gamma,1} \in \mathcal{F}_{1,\theta_1}$ when $\|\gamma\| \leq r_0$, we may apply~\eqref{Eq:GenHolder}, the second part of \citet[][Proposition~9]{BSY2017},~\eqref{Eq:ii,1},~\eqref{Eq:ii,2} and the fact that $\alpha_2 \geq 4$ to deduce that 
\begin{align}
\label{Eq:hgamma}
	\sup_{f \in \mathcal{F}_{p+1,\omega}^*} &\sup_{\gamma \in B_0^\circ(r_0)} \frac{|h(\gamma)-h(0)|}{\|\gamma\|}  \nonumber \\
&\leq \sup_{f \in \mathcal{F}_{p+1,\omega}^*} \sup_{\gamma \in B_0^\circ(r_0)} \frac{1}{\|\gamma\|} \max \biggl\{ \mathbb{E} \biggl| \log \frac{ f_{\eta}(\eta_1 - X_1^T\gamma)}{ f_{\eta}(\eta_1)} \biggr| \, , \, \mathbb{E} \biggl| \log \frac{ f_{\hat{\eta}}^{\gamma,1}(\eta_1 - X_1^T\gamma)}{ f_{\hat{\eta}}^{\gamma,1}(\eta_1)} \biggr| \biggr\} < \infty.
\end{align}
Similarly, 
\begin{equation}
\label{Eq:hgamma1}
\sup_{f \in \mathcal{F}_{p+1,\omega}^*} \sup_{\gamma \in B_0^\circ(r_0)} \frac{|h^{(1)}(\gamma)-h^{(1)}(0)|}{\|\gamma\|} < \infty.
\end{equation}
We now study the convergence of $\hat{\gamma}$ and $\hat{\gamma}^{(1)}$ to $0$. By definition, the unique minimiser of the function
\[
	R(\beta):= \mathbb{E}\{(Y_1-X_1^T\beta)^2\} = \mathbb{E}(\epsilon_1^2) - 2(\beta-\beta_0)^T \mathbb{E}(\epsilon_1X_1) + (\beta-\beta_0)^T \Sigma(\beta-\beta_0)
\]
is given by $\beta_0$, and $R(\beta_0)=\mathbb{E}(\epsilon_1^2)$.  Now, for $\lambda>0$, 
\begin{align*}
	R\bigl(\beta_0+\lambda \mathbb{E}(\epsilon_1X_1)\bigr) = \mathbb{E}(\epsilon_1^2) - 2 \lambda \|\mathbb{E}(\epsilon_1X_1)\|^2 + \lambda^2 \mathbb{E}(\epsilon_1X_1)^T \Sigma \mathbb{E}(\epsilon_1X_1).
\end{align*}
If $\mathbb{E}(\epsilon_1X_1) \neq 0$ then taking $\lambda = \frac{\|\mathbb{E}(\epsilon_1X_1)\|^2}{\mathbb{E}(\epsilon_1X_1)^T \Sigma \mathbb{E}(\epsilon_1X_1)}$ we have $R\bigl(\beta_0+\lambda \mathbb{E}(\epsilon_1X_1)\bigr) < \mathbb{E}(\epsilon_1^2)$, a contradition.  Hence $\mathbb{E}(\epsilon_1X_1) =0$.  Moreover,
\begin{align}
\label{Eq:SecondTermBound}
	\mathbb{E} \Bigl\| \frac{1}{m} \sum_{i=m+1}^n X_i \eta_i \Bigr\| &\leq \frac{1}{m} \sum_{j=1}^p \mathbb{E} \Bigl| \sum_{i=m+1}^n X_{ij}\eta_i \Bigr| \leq \frac{1}{m} \sum_{j=1}^p \mathrm{Var}^{1/2} \Bigl( \sum_{i=m+1}^n X_{ij} \eta_i \Bigr) \nonumber \\
&= \frac{1}{m^{1/2}} \sum_{j=1}^p \mathrm{Var}^{1/2} (X_{1j}\eta_1) \leq \frac{2^{1/2}p}{n^{1/2}} \mathbb{E}(\eta_1^4)^{1/4} \max_{j=1,\ldots,p} \mathbb{E}(X_{1j}^4)^{1/4}.
\end{align}
It follows that 
\begin{equation}
\label{Eq:Error}
	\mathbb{E}\bigl(\|\hat{\gamma}\|\mathbbm{1}_{A_0}\bigr) = \mathbb{E}\Bigl\{\bigl\|(X_{(2)}^TX_{(2)})^{-1}X_{(2)}^T \eta_{(2)}^{(0)}\bigr\|\mathbbm{1}_{A_0}\Bigr\} \leq \frac{2^{1/2}p}{\lambda_0n^{1/2}} \mathbb{E}(\eta_1^4)^{1/4} \max_{j=1,\ldots,p} \mathbb{E}(X_{1j}^4)^{1/4}.
\end{equation}
Similar arguments yield the same bound for $\mathbb{E}\bigl(\|\hat{\gamma}^{(1)}\|\mathbbm{1}_{A_0}\bigr)$.  Hence, from~\eqref{Eq:hgamma},~\eqref{Eq:hgamma1} and~\eqref{Eq:Error}, we deduce that
\begin{equation}
\label{Eq:R2Bound}
\limsup_{n \rightarrow \infty} \ n^{1/2} \sup_{\substack{k_\eta \in \{k_0^*,\ldots,k_\eta^*\}\\k \in \{k_0^*,\ldots,k^*\}}} \sup_{f \in \mathcal{F}_{p+1,\omega}^*} \mathbb{E}_f(|R_2|\mathbbm{1}_{A_{r_0,s_0} \cap A_0}) < \infty.
\end{equation}
We now bound $\mathbb{P}(A_0^c)$.  By the Hoffman--Wielandt inequality, we have that
\begin{align*}
	\{\lambda_\text{min}(m^{-1}X_{(2)}^TX_{(2)}) - \lambda_\text{min}(\Sigma)\}^2 \leq \|m^{-1}X_{(2)}^TX_{(2)} - \Sigma \|_\mathrm{F}^2.
\end{align*}
Thus
\begin{align}
\label{Eq:A0c}
	\mathbb{P}(A_0^c) \leq \mathbb{P} \Bigl( \|m^{-1}X_{(2)}^TX_{(2)} - \Sigma\|_\mathrm{F} \geq \frac{1}{2} \lambda_\text{min}(\Sigma) \Bigr) &\leq \frac{4}{m\lambda_\text{min}^2(\Sigma)} \sum_{j,l=1}^p \mathrm{Var}(X_{1j}X_{1l}) \nonumber \\
	& \leq \frac{8p^2}{n\lambda_\text{min}^2(\Sigma)} \max_{j=1,\ldots,p} \mathbb{E}(X_{1j}^4).
\end{align}
Finally, we bound $\mathbb{P}(A_0 \cap A_{r_0,s_0}^c)$.  By Markov's inequality and~\eqref{Eq:Error},
\begin{equation}
\label{Eq:gammahatbound}
\mathbb{P}\bigl(\{\|\hat{\gamma}\| \geq r_0\} \cap A_0\bigr) \leq \frac{1}{r_0}\mathbb{E}\bigl(\|\hat{\gamma}\|\mathbbm{1}_{A_0}\bigr) \leq \frac{2^{1/2}p}{r_0\lambda_0n^{1/2}} \mathbb{E}(\eta_1^4)^{1/4} \max_{j=1,\ldots,p} \mathbb{E}(X_{1j}^4)^{1/4}.
\end{equation}
The same bound also holds for $\mathbb{P}\bigl(\{\|\hat{\gamma}^{(1)}\| \geq r_0\} \cap A_0\bigr)$.  Furthermore, writing $P_{(2)} := X_{(2)}(X_{(2)}^TX_{(2)})^{-1}X_{(2)}^T$, note that
\begin{align*}
	|\hat{s}^2 - 1| = \Bigl| \frac{1}{m}\|\eta_{(2)}^{(0)}\|^2 - 1 - \frac{1}{m}\|P_{(2)}\eta_{(2)}^{(0)}\|^2\Bigr| \leq \frac{1}{m}\biggl|\sum_{i=m+1}^n (\eta_i^2 - 1)\biggr| + \frac{1}{m} \eta_{(2)}^{(0)T}P_{(2)}\eta_{(2)}^{(0)},
\end{align*}
so by Chebychev's inequality, Markov's inequality and~\eqref{Eq:SecondTermBound}, for any $\delta > 0$
\begin{align}
\label{Eq:ssquared}
\mathbb{P}\bigl(\{|&\hat{s}^2 - 1| > \delta\} \cap A_0\bigr) \nonumber \\
&\leq \mathbb{P}\biggl(\biggl\{\frac{1}{m}\biggl|\sum_{i=m+1}^n (\eta_i^2 - 1)\biggr| > \frac{\delta}{2}\biggr\} \cap A_0\biggr) + \mathbb{P}\biggl(\biggl\{\biggl\|\frac{1}{m}X_{(2)}^T\eta_{(2)}^{(0)}\biggr\| > \Bigl(\frac{\lambda_0\delta}{2}\Bigr)^{1/2}\biggr\} \cap A_0\biggr) \nonumber \\
&\leq \frac{8}{n\delta^2}\mathbb{E}(\eta_1^4) + \frac{2p}{\delta^{1/2}\lambda_0^{1/2}n^{1/2}} \mathbb{E}(\eta_1^4)^{1/4} \max_{j=1,\ldots,p} \mathbb{E}(X_{1j}^4)^{1/4}.
\end{align}
The same bound holds for $\mathbb{P}\bigl(\{|(\hat{s}^{(1)})^2 - 1| > \delta\} \cap A_0\bigr)$.  We conclude from Markov's inequality,~\eqref{Eq:BreveDecomp},~\eqref{Eq:R1R2},~\eqref{Eq:R1Bound},~\eqref{Eq:R2Bound},~\eqref{Eq:A0c},~\eqref{Eq:gammahatbound} and~\eqref{Eq:ssquared} that
\begin{align*}
\sup_{B_n \geq B_n^*} &\sup_{\substack{k_\eta \in \{k_0^*, \ldots, k_\eta^*\} \\ k \in \{k_0^*, \ldots, k^*\}}}\sup_{f \in \mathcal{F}_{1,p+1}^*:I(f) \geq b_n} \mathbb{P}_f(\breve{I}_n \leq \breve{C}_q^{(n),B_n}) \\
&\leq \frac{1}{q(B_n^*+1)} + \sup_{\substack{k_\eta \in \{k_0^*, \ldots, k_\eta^*\} \\ k \in \{k_0^*, \ldots, k^*\}}}\sup_{f \in \mathcal{F}_{1,p+1}^*:I(f) \geq b_n}\mathbb{P}_f(\breve{I}_n^{(1)} \geq \breve{I}_n^{(0)}) \\
&\leq \frac{1}{q(B_n^*+1)} + \sup_{\substack{k_\eta \in \{k_0^*, \ldots, k_\eta^*\} \\ k \in \{k_0^*, \ldots, k^*\}}}\sup_{f \in \mathcal{F}_{1,p+1}^*:I(f) \geq b_n} \Bigl\{\mathbb{P}_f\bigl(\{\breve{I}_n^{(1)} \geq \breve{I}_n^{(0)}\} \cap A_{r_0,s_0} \cap A_0\bigr) \\
&\hspace{5cm} + \mathbb{P}_f(A_0 \cap A_{r_0,s_0}^c) + \mathbb{P}_f(A_0^c)\Bigr\} \\
&\leq \frac{1}{q(B_n^*+1)} + \sup_{\substack{k_\eta \in \{k_0^*, \ldots, k_\eta^*\} \\ k \in \{k_0^*, \ldots, k^*\}}}\sup_{f \in \mathcal{F}_{1,p+1}^*:I(f) \geq b_n} \Bigl\{\frac{1}{b_n}\mathbb{E}_f(|R_1+R_2|\mathbbm{1}_{A_{r_0,s_0}})  \\
&\hspace{5cm} + \mathbb{P}_f(A_0 \cap A_{r_0,s_0}^c) + \mathbb{P}_f(A_0^c)\Bigr\} \rightarrow 0,
\end{align*}
as required.
\end{proof}

\textbf{Acknowledgements}: Both authors are supported by an EPSRC Programme grant.  The first author was supported by a PhD scholarship from the SIMS fund; the second author is supported by an EPSRC Fellowship and a grant from the Leverhulme Trust.



\end{document}